\documentclass{article} % use larger type; default would be 10pt

\usepackage[utf8]{inputenc} % set input encoding (not needed with XeLaTeX)
\usepackage{color}

\usepackage{geometry} % to change the page dimensions
\usepackage{graphicx} % support the \includegraphics command and options

\usepackage{booktabs} % for much better looking tables
\usepackage{array} % for better arrays (eg matrices) in maths
\usepackage{paralist} % very flexible & customisable lists (eg. enumerate/itemize, etc.)
\usepackage{verbatim} % adds environment for commenting out blocks of text & for better verbatim
\usepackage{subfig} % make it possible to include more than one captioned figure/table in a single float
\usepackage{fancyhdr, url,graphicx,tabularx,array,geometry,listings, amsmath, multirow, amsfonts, enumerate, mathabx, bbm} % Required for custom headers
\usepackage{lastpage} % Required to determine the last page for the footer
\usepackage{extramarks} % Required for headers and footers
\usepackage{graphicx} % Required to insert images
\usepackage{lipsum} % Used for inserting dummy 'Lorem ipsum' text into the template
\usepackage{color}
\usepackage{empheq}
\usepackage{mathtools}
\usepackage{amssymb}
\usepackage{mdframed}
 \usepackage{hyperref}
\usepackage{amsthm}
\newtheorem{theorem}{Theorem}
\newtheorem{lemma}[theorem]{Lemma}
\newtheorem{corollary}{Corollary}[theorem]
\newtheorem{definition}{Definition}[section]
\newtheorem{example}{Example}[section]
\newtheorem*{remark}{Remark}
\usepackage{fancyhdr} % This should be set AFTER setting up the page geometry
\usepackage{sectsty}
\allsectionsfont{\sffamily\mdseries\upshape} % (See the fntguide.pdf for font help)
\usepackage[nottoc,notlof,notlot]{tocbibind} % Put the bibliography in the ToC
\usepackage[titles,subfigure]{tocloft} % Alter the style of the Table of Contents

\title{On the Gaussianity of Kolmogorov Complexity of Mixing Sequences}
\author{Morgane Austern, Arian Maleki}
\date{}
\begin{document}
\maketitle

\begin{abstract}
 Let $ K(X_1, \ldots, X_n)$ and $H(X_n | X_{n-1}, \ldots, X_1)$ denote the Kolmogorov complexity and Shannon's entropy rate of a stationary and ergodic process $\{X_i\}_{i=-\infty}^\infty$. It has been proved that
\[
\frac{K(X_1, \ldots, X_n)}{n} - H(X_n | X_{n-1}, \ldots, X_1) \rightarrow 0,
\]
almost surely. This paper studies the convergence rate of this asymptotic result. In particular, we show that if the process satisfies certain mixing conditions, then there exists $\sigma<\infty$ such that
$$\sqrt{n}\left(\frac{K(X_{1:n})}{n}-  H(X_0|X_1,\dots,X_{-\infty})\right) \rightarrow_d N(0,\sigma^2).$$
Furthermore, we show that under slightly stronger mixing conditions one may obtain non-asymptotic concentration bounds for the Kolmogorov complexity.
\end{abstract}

\section{Introduction}

\subsection{Motivation and objective}
Kolmogorov complexity of a binary sequence is defined as the length of the shortest program fed to a universal Turing machine that would print the sequence and halt. More formally, let $U$ denote a Universal Turing machine. Given a program $p$ the sequence printed by $U$ is denoted with $U(p)$. 
\begin{definition}\label{Kolmogorov complexity} 
Let $\mathcal{P}_X$ denote the set of all binary programs that can generate a finite length binary sequence $X$ and halt. Then, the Kolmogorov complexity of $X$ is denoted with $K(X)$ and is defined as $$K(X) \triangleq \inf_{p \in \mathcal{P}_X} {\rm length}(p),$$ where ${\rm length}(p)$ denotes the length of the sequence. Furthermore, the Kolmogorov complexity of any finite-length finite-alphabet sequence is the Kolmogorov complexity of its binary representation.
\end{definition}
 Apart from its mathematical elegance, Kolmogorov complexity has exhibited promising {\em theoretical} results in other areas of research including inductive inference \cite{solomonoff1978complexity}, denoising \cite{donoho2002kolmogorov}, linear regression \cite{jalali2011minimum}, density estimation \cite{barron1991minimum}, etc. However, such theoretical results are overshadowed by the fact that Kolmogorov complexity is not computable (see Theorem 1.5 in \cite{zvonkin1970complexity}). 

 Both the usefulness of the Kolmogorov complexity and its incomputability has motivated researchers to find approximations of this quantity. One of the main approaches is to restrict the class of sequences to stationary and ergodic sequences, and use the properties of such sequences to find good approximations. The following theorem, due to Levin, clarifies why such assumptions might be useful.

\begin{theorem}\cite{zvonkin1970complexity} \label{thm:first_order}
Let $\{X_i\}_{i = -\infty}^{\infty}$ denote a binary stationary and ergodic process (under left shift), whose law is computable. Let its Shannon conditional entropy rate be $H(X_1|X_0,\dots,X_{-\infty})$. Then,\footnote{This is Theorem 5.1 of \cite{zvonkin1970complexity}. As mentioned by Levin even ergodicity is not necessary, but then we should be careful in defining the entropy. For more information refer to \cite{zvonkin1970complexity}. }
\begin{equation*}
\frac{K(X_{1:n})}{n} \overset{a.s.}{\rightarrow} H(X_1|X_0,\dots,X_{-\infty}), 
\end{equation*}
where $X_{1:n}$ denotes the vector $(X_1, X_2, \ldots, X_n)$. 
\end{theorem} 
According to this theorem Shannon's entropy can be seen  as an approximation of the Kolmogorov's complexity of the process. This result is asymptotic and it is not clear how accurate this approximations is even if $n$ is large. This paper establishes the accuracy of this approximation under certain mixing assumptions on the process, which will be clarified later. 

\subsection{Related work}
Komogorov complexity evolved in the seminal papers of Solomonoff \cite{solomonoff1964formal1, solomonoff1964formal2}, Kolmogorov \cite{kolmogorov1965three}, and Chaitin \cite{chaitin1969simplicity, chaitin1966length, chaitin1975theory, chaitin1974information}. Each author developed and used this quantity for different purposes. For instance, inspired by Shannon's theory of information, Kolmogorov developed his notion of complexity to quantify the amount of information that is present in a sequence of bits. Kolmogorov also conjectured that binary sequences that have maximal complexity, e.g. $K(X_1, X_2, \ldots, X_n) \geq n-c$, for some fixed $c$, are random in an intuitive sense.This conjecture was later established by Martin-Lof. Intuitively speaking he proved that if a sequence satisfies $K(X_1, X_2, \ldots, X_n) \geq n-c$, then any test that can be implemented by a turing machine will accept the randomness of this sequence (it should possibly use a different significance level). The intellectual value of this test of randomness was overshadowed by the incomputability of Kolmogorov complexity. 

Many researchers explored new ways to improve the applicability of Kolmogorov complexity. For instance,   \cite{willis1970computational, chaitin1974information, hartmanis1983generalized, sipser1983complexity} explored computable approximations of Kolmogorov complexity. Another popular direction of research pursued connections between Kolmogorov's complexity and Shannon entropy \cite{kolmogorov1965three, zvonkin1970complexity, v1998ergodic, brudno1982entropy}. Levin's result, i.e. Theorem \ref{thm:first_order},  is one of the most general connections between Kolmogorov complexity and Shannon entropy. In this paper, we push these connections one step further by providing convergence rate and concentration results for the Kolmogorov complexity.

\section{Main result}\label{sec:mainresult}

%Let $\{\Omega, \mathbb{P},\mathcal{A}\}$ be a probability space, \textcolor{green}{Do we want our process to be only real}

 According to Theorem \ref{thm:first_order}, for every stationary and ergodic sequence, $\{X_i\}_{i=-\infty}^{\infty}$, on  probability space $\{\Omega, \mathbb{P},\mathcal{A}\}$, we have  
\begin{equation*}
\frac{K(X_{1:n})}{n} \overset{a.s.}{\rightarrow} H(X_1|X_0,\dots,X_{-\infty}). 
\end{equation*}
As we discussed before, in this paper we would like to characterize the rate of convergence for this asymptotic result. As our first goal, we would like to show that under some general conditions the convergence rate is $\frac{1}{\sqrt{n}}$. More specifically, we would like to show that 
\begin{equation}\label{eq:CLTtypeKC}
\sqrt{n} \left(\frac{K(X_{1:n})}{n} -H(X_1|X_0,\dots,X_{-\infty})  \right)
\end{equation}
converges in distribution to a non-degenerate random variable. Before we discuss our main theorem we would like to show that if we do not impose extra conditions on the process, the convergence rate could be slower than  $\frac{1}{\sqrt{n}}$. 
%
%\begin{example}[Faster convergence rate] \label{ex:nonmixingexample}
%Let $\{X_i\}_{i=0}^{\infty}$ denote a stochastic process that is constructed in the following way: (i)$\{X_i\}_{i=0}^\infty$ is a first order Markov process. (ii) $X_0$ is a Bernoulli(0.5) random variable. (iii) For every $i \geq 0$ we have $X_{i+1}=\begin{cases} 1  \mbox{ \rm if } X_i=0 \\ 0 \ \ X_i=1\end{cases}$. 
%
%First note that $\{X_i\}_{i=0}^\infty$ is ergodic for the left shift. Also, it is straightforward to check that $H(X_i \ | \ X_{0}^{i-1}) =0$. We present an upper bound on the Kolmogorov complexity. Note that if we tell the first bit to the Universal computer and then  ask the computer to alternate between zero and one the universal computer can figure out the sequence. Hence,
%$$\frac{K(X_{1:n})}{n} \le \frac{C +K(X_1)}{n}=\frac{\tilde{C} +1}{n}. $$
%It is also straightforward to see that since there are two distinct sequences, we require at least one bit to describe them, hence the rate of convergence to zero is exactly $1/n$.  
% \end{example}
%  
%Note that in the previous case the rate of convergence is faster than $1/\sqrt{n}$. We can also show some examples in which the convergence rate is slower than $1/\sqrt{n}$.

\begin{example}\label{ex:slower}
Let $\nu_1$ denote a probability measure  on $\mathbb{N}^*$ (the set of  positive natural numbers) with probability mass function $\nu_1(t)=  \frac{Z }{  t^{3/2-\epsilon}}$, where $\epsilon\in (0,\frac{1}{6})$ and Z is the normalizing constant. Let $\{\tau_i\}_{i=1}^\infty$ denote iid samples from this distribution. Furthermore, let $\{Y_i\}_{i=1}^\infty$ and $\{\tilde{Y}_i\}_{i=1}^\infty$ denote two independent sequences of  iid $\rm{Bern}( \frac{1}{2})$ random variables ( independent from $\{\tau_i\}_{i=1}^\infty$).  Given a natural number $a$ let $\mathbf{0}_a$ denote a vector of size $a$ with all elements being zero. We construct a binary sequence $\{\tilde{X}_i\}_{i=1}^\infty$ in the following way:

\begin{itemize}
\item[(i)]  Pick $\tau_1$ from our first sequence and set $\tilde{X}_{1: \tau_1} = Y_1 \mathbf{0}_{\tau_1} + (1-Y_1) \tilde{Y}_{1: \tau_1}$. 

\item[(ii)] To construct the $i^{\rm th}$ block we repeat what we did above. More specifically we draw $\tau_i$ and we set $\tilde{X}_{\tau_1+\ldots+ \tau_{i-1}:\tau_1+\ldots+ \tau_{i-1}+\tau_i} = Y_i \mathbf{0}_{\tau_i} + (1-Y_i) \tilde{Y}_{\tau_1+\ldots+ \tau_{i-1}:\tau_1+\ldots+ \tau_{i-1}+\tau_i}$.
\end{itemize}
 To make $\tilde{X} = \{\tilde{X}\}_{i=1}^{\infty}$ stationary, draw $\theta | \tau_1 \sim \rm{unif}(0,\tau_1-1)$ (uniform on the integers from $0$ to $\tau_1-1$). Then, we generate the new process as  $X=\Theta^{\theta} \tilde{X}$, where $\Theta$ is the left shift operator. It is straightforward to see that the process is stationary and ergodic.  Hence, $\frac{K(X_{1:n})}{n} \overset{a.s.}{\rightarrow} H(X_1|X_0,\dots,X_{-\infty})$. However, the convergence rate is slower than $1/\sqrt{n}$. The proof of our claim can be found in Section \ref{sec:proofexampleslow}. 
 \end{example}

Note that a major issue in the above example is the fact that the elements of the sequence that are far apart can still have strong dependencies. Hence, intuitively speaking we expect that if the dependency of  the process is weaker, then we may be able to obtain the $1/\sqrt{n}$ convergence rate. Mixing conditions are defined to capture the dependancies of stochastic process. We start with  mixing conditions that will be used in our paper. Let $\{X_i\}_{i=-\infty}^{\infty}$ denote a stationary and ergodic process, and let  $\mathbb{F}_{-\infty}^n= \sigma(X_i, i \le n)$ denote the $\sigma$-filed of events generated by random variables $\ldots, X_{n-2}, X_{n-1}, X_{n}$  and $\mathbb{F}_{n}^{\infty}= \sigma(X_i, i \ge n)$ denote the $\sigma$-filed of events generated by $X_n, X_{n+1}, \ldots$.

\begin{definition} $\alpha$-mixing coefficients of the process $\{X_i\}_{i=-\infty}^{\infty}$ is define as
 $$\alpha(n) \triangleq \sup_j  \sup_{\substack{A \in \mathbb{F}_{-\infty}^j \\ B  \in \mathbb{F}_{j+n}^{\infty} }} \left|P\left(A \bigcap B\right) -P(A)P(B) \right|.$$ 
A process $X$ is $\alpha$-mixing if $\alpha(n) \rightarrow 0.$
\end{definition}

$\alpha$-mixing condition ensures that the parts of the process that are far apart are almost independent. Hence, we hope that if $\alpha(n)$ decays fast enough, then it will avoid the dependency issue raised in Example \ref{ex:slower}. 
In some of our results we will need a slightly stronger notion of mixing that we define below.

\begin{definition}
The $\phi$-mixing coefficient of the process $\{X_i\}_{i=-\infty}^{\infty}$ is defined as
 $$\phi(n) \triangleq \sup_j  \sup_{\substack{A \in \mathbb{F}_{-\infty}^j \\ B  \in \mathbb{F}_{j+n}^{\infty} }} \left|P\left(B | A\right) -P(B) \right|,$$ 

Furthermore, a process is called $\phi$-mixing  if $\phi(n) \rightarrow 0.$

\end{definition}
\begin{remark}
It is straightforward to see that $\forall n$, $\phi(n) \ge \alpha(n)$. Hence, if a process is $\phi$-mixing, then it will be $\alpha$-mixing. 
\end{remark}

In addition to mixing, our proof requires another condition that is described below:

\begin{definition}\label{def:nudelta}
Let $\{X_i\}_{-\infty}^{\infty}$ denote a stationary and ergodic process. Consider $\delta>0$ and define 
\vspace{-.2cm}
$$\nu_{\delta}(n)\triangleq \mathbb{E}(|\log (P(X_0|X_{-1},\dots,X_{-\infty}))- \log (P(X_0|X_{-1},\dots,X_{-n}))|^{\frac{2+\delta}{1+\delta}}).$$
\end{definition}
Note that in this paper all the logarithms are  in base 2. Since $\nu_{\delta}(n)$ is not a standard notion in probability theory, we explain some of its interesting features below:

\begin{enumerate}

\item The definition of $\nu_\delta(n)$ is close to the definition of the Kullback-Leiber divergence between $P(X_0 | X_{-1}, \ldots, X_{-\infty})$ and $P(X_0 \ | \ X_{-1}, \ldots, X_{-n})$. Hence, it measures the discrepancy of a process from a Markov process. 

\item For a $b$-Markov source $\{X_i\}_{i=-\infty}^{\infty}$ we have $\mathcal{L}(X_1|X_{0:-b+1}) = \mathcal{L}(X_1|X_{0:-\infty}) $\footnote{ Where $\mathcal{L}(X_1|X_{0:-b+1}) $ is the conditional distribution knowing $X_{0:-b+1}$ of $X_1$, and $ \mathcal{L}(X_1|X_{0:-\infty}) $ is the conditional distribution knowing $X_{0:-\infty}$ of $X_1$}. Hence $\nu_\delta(n) =0$ for every $n \ge b$. 

\item Sequences generated by a hidden Markov model also have very fast decaying $\nu_\delta(n)$. The following lemma justifies our claim:
\begin{lemma}\label{lem:HMMmixing}
Consider a hidden Markov model with $q: \mathcal{X} \times \mathcal{X} \rightarrow (0, \infty)$ denoting the transition kernel of the underlying Markov process and  $g(\cdot \ | \ x)$ denoting the distribution of the observed variables for a given value $x$ of the hidden variable. Also, suppose that the process satisfies the following conditions: 
\begin{enumerate}
\item[(i)]  $\epsilon\triangleq\frac{{\rm essinf}\ q(x,x')}{{\rm esssup} \ q(x,x')} \in (0,1)$. 
\item[(ii)] $ 1<\eta \triangleq \sup_y \frac{{\rm esssup}_x g(y|x)}{{\rm essinf}_x g(y|x)} < \infty$.
%\item[(iii)] $\eta'\triangleq {\rm essinf_x} \ g(y|x) >0.$ \textcolor{red}{Does it depend on $y$? or you take max over $y$? }
\end{enumerate}
Then, if $Y_1, Y_2, \ldots$ is the sequence of observations generated by this process there exists a value of $\tau \in (0,1)$ only depending on $\epsilon$, $\delta$, and $\eta,$ such that
\[
\nu_\delta (n) = \mathbb{E}\left(\left|\log\left(\frac{P(Y_0 |Y_{-1:-n})}{P(Y_0 |Y_{-1:-\infty})}\right)\right|^{\frac{2+\delta} {1+\delta}}\right) \leq   C \tau^{n},
\]
where C is a constant that depends only on $\epsilon$ and $\eta$. 
\end{lemma}
The proof of this lemma is presented in Section \ref{sec:lemHMMproof}. 
\end{enumerate}

% For $\{X_i\}_{i=-\infty}^{\infty}$ a finite-state m-markov chain, there is $M<\infty$ such that, by changing coordinate of $X_{1:n}$, almost-surely we don't modify $\log P(X_{1:n})$  by more than M.We will need such notions in some of our results. 

In addition to the above mixing conditions, we require a notion of stability for the likelihood of a process for our finite sample concentration results. To understand this notion we should first define the Hamming distance between two vectors.

\begin{definition}
The Hamming distance between two sequences $x_{1:n}  \in \mathbb{R}^n$ and  $y_{1:n} \in \mathbb{R}^n$ is defined as $$d_n(x_{1:n}, y_{1:n}) \triangleq \sum_{i=1}^n \mathbb{I}_{x_i\neq y_i},$$
where $\mathbb{I}$ denotes the indicator function. 
\end{definition}

This notion enables us to define the notion of $M-$stability.
\begin{definition}
The $M$-stability  coefficient of a finite state m-Markov process $\{X_i\}_{i=-\infty}^{\infty}$, with $X_i \in A$, is defined as
$$ M \triangleq \sup_{n} \sup_{(X_{1:n},X'_{1:n})\in {A^n}^2,~{\rm s.t.}~d_n(X_{1:n},X'_{1:n})\le 1 } |\log(P(X_{1:n})- \log(P(X'_{1:n}))|.$$ 
We will say that $\{X_i\}_{i=-\infty}^{\infty}$ is $M$-stable if its $M-$stability coefficient is finite.
%More generaly  for a process $\{X_i\}_{i=-\infty}^{\infty}$ defined on a finite-alphabet space. It's $M-$stability coefficient will be defined by:
%$$ M \triangleq \sup_m \sup_{n} \sup_{(x_{1:n},x'_{1:n})\in {A^n}^2,~\rm{s.t}~ d_n(x_{1:n},x'_{1:n})\le 1 } |\log(P_{m}(x_{1:n})- \log(P_{m}((x'_{1:n})))|,$$ where $P_{m}$ is the distribution of a $m-$markov chain with the same distribution for  $X_{1:-m}$ than $\{X_i\}_{i=-\infty}^{\infty}$.
%
%We will say that a process $\{X_i\}_{i=-\infty}^{\infty}$ is $M-$stable, if $M<\infty.$
\end{definition}

%\begin{remark}
%%Note that for a $m-$markov finite state markov $\{X_i\}_{i=-\infty}^{\infty}$, such that $\rho \triangleq %\min_{x_{1:-m} \in A^{m+2}}  P(x_1|x_{0:-m})>0$ then: $M< \infty$.
%
%In general for a process  $\{X_i\}_{i=-\infty}^{\infty}$ such that:
%\begin{itemize}
%\item $\inf_{\{x_i\}_{i=-\infty}^{\infty} \in A^{\mathbb{Z}}} \{P(x_1|x_{0:-\infty})\}>0$.
%\item $\exists p_k$, such as $\forall \{x_i\}_{i=-\infty}^{\infty} \in A^{\mathbb{Z} }~|\log (P(x_1|x_{0:-\infty}))-\log (P(x_1|x_{0:-k}))| \le p_k $ and $\sum_{k=0}^{\infty} p_k<\infty$.
%\end{itemize}
%Then $M<\infty.$
%A proof of this will be found in the proof section. \textcolor{red}{A raoujouter}

\begin{remark}\label{upper_bound_M}
Consider a finite-state $m-$Markov chain $\{X_i\}_{i=-\infty}^{\infty}$. If 
$\rho \triangleq \min_{x_{1:-m} \in A^{m+2}}   P(x_1|x_{-m:0})>0$, then the M-stability coefficient satisfies
$$M\le (m+1)\log\Big(\frac{1}{\rho}\Big).$$
 The proof of this claim is presented in Section \ref{sec:proofremarkMstable}. 

The notion of $M-$stability  will be used to obtain finite-sample concentration results.  This notion can be seen in relation with the vast majority of concentration inequalities, such as Azuma, Hoefding, and McMiarmid that require boundedness conditions.  

\end{remark}

Now using the notions we developed above we state our first main result that confirms the asymptotic Gaussianity of the Kolmogorov complexity of ergodic sequences. 

\begin{theorem}\label{thm:quantized}
Let  $\{X_i\}_{-\infty}^{\infty}$ denote a stationary and ergodic process. We assume that $X_1 \in A$, where $A=\{a_1,..,a_l\}$ with $l < \infty$. Furthermore, we suppose that

\begin{enumerate}

\item[$C_1.$] The Kolmogorov complexity of  all $a_j $s is finite, i.e., $\max_{i \in \{1, \ldots, l \}}K(a_i)< \infty$.

%\item [$C_2.$]  $H(X_0|X_{-1},...,X_{-\infty})>0 $. 

\item [$C_2.$] 
We assume that there are fixed numbers  $K$, $\beta >1, C>1$, and $\delta \in (0,1]$, such that
\begin{itemize}
\item[-]$\alpha(n) \le K n^{-\beta \frac{(2+\delta) (1+\delta)}{\delta^2}}.$
\item[-] ${\nu_{\delta}(n)}^{\frac{1+\delta}{2+\delta}} = O\left( 2^{-Cn \log(l)} \right).$
\end{itemize} 

\end{enumerate}
If we define {\small \begin{equation*} \begin{split} \sigma^2 & \triangleq {\rm var} (\log(P(X_0|X_{-1},\dots,X_{-\infty})))  + 2 \sum_{k} {\rm cov}(\log(P(X_0|X_{-1},\dots,X_{-\infty})),\log(P(X_k|X_{k-1},\dots,X_{-\infty}))), \end{split} \end{equation*}}
then $\sigma^2 < \infty$, and $$\sqrt{n}\left(\frac{K(X_{1:n})}{n}-  H(X_0|X_{-1},\dots,X_{-\infty})\right) \rightarrow_d N(0,\sigma^2),$$
where the notation $\rightarrow_d$ is used for the convergence in distribution.
\end{theorem}
 The proof of this theorem is presented in Section \ref{sec:prooffinitecase}. Note that Theorem \ref{thm:quantized} implies Theorem \ref{thm:first_order}. However, this result provides the rate of convergence as well. Both Theorem \ref{thm:first_order} and Theorem \ref{thm:quantized} are  concerned with the asymptotic behavior of the Kolmogorov complexity, and do not provide any information on the finite sample behavior of this quantity. The following corollary simplifies the statement of this theorem for an independent and identically distributed sequence.

\begin{corollary}
Let  $\{X_i\}_{-\infty}^{\infty}$ denote an independent and identically distributed process. We assume that $X_1 \in A$, where $A=\{a_1,..,a_l\}$ with $l < \infty$. Furthermore, we assume that  $\max_{i \in \{1, \ldots, l \}}K(a_i)< \infty$. Then,
$$\sqrt{n}\left(\frac{K(X_{1:n})}{n}-  H(X_0) \right) \rightarrow_d N(0,\sigma^2),$$
where $\sigma^2= {\rm var} (\log(P(X_0)))$. 
\end{corollary}

Our next goal is to derive probabilistic upper bounds on the discrepancy of the Kolmogorov complexity and Shannon entropy in finite sample sizes. Our next theorem shows that such bounds can be obtained with slightly stronger mixing conditions than those in Theorem \ref{thm:quantized}. For an integer number $n$ define
\begin{eqnarray}
 \log^*(n) = \left\{ \begin{array}{ll}
0 & n \leq 1 \\
 1+ \log^*(\log(n))  & n>1.
\end{array}
\right.
\end{eqnarray}

\begin{theorem}\label{thm:expo non-markov}
Let  $\{X_i\}_{-\infty}^{\infty}$ denote a stationary $m$-Markov process. We assume that $X_1 \in A$, where $A=\{a_1,..,a_l\}$ with $l < \infty$. Furthermore, we assume that

\begin{enumerate}

\item The Kolmogorov complexity of  all $a_j $s is finite, i.e., $\max_{i \in \{1, \ldots, l \}}K(a_i)< \infty$.

%\item $H(X_0|X_{-1},...,X_{-m})>0 $. 

\item The $M$-stability coefficient of the process, M, is finite.
\item   The $\phi$-mixing coefficients of the process satisfy $\Delta \triangleq 1 + 24 \sum_{k=0}^{\infty} \phi(k)< \infty $.
%\item $\alpha(n) \rightarrow 0.$

%\item $\exists \delta>0$ $\exists C>1$ such that $v_{\delta}(n)^{\frac{1+\delta}{2+\delta}}=O(2^{-Cn\log(l)})$.
\end{enumerate}

Let $\eta \in (0,0.5)$ be a fixed number. We will have $C'$ a constant the depends only on the universal mahcine, and define $\gamma_n \triangleq  \frac{C_1(n)}{n}+ \frac{m}{n} H(X_{1}|X_{0:-m+1})+ n^{-\frac{1}{2}-\eta}=O( n^{-\frac{1}{2}-\eta})$, where $C_1(n)\triangleq  C' + \log^*(m) + l \max_{j \le l} K(a_j)  + l^{(m+1)} \log ^* n  - m \log^*l $. Moreover have $\gamma'(n)\triangleq C' + \log^*(m) + l \max_{j \le l} K(a_j)  + l^{m+1} \log ^* n  + m \log^*l +m H(X_1)=O(\frac{\log^*(n)}{n})$, $ K_1 = 2M^2 \Delta^2$ and  $K'_1(n)\triangleq 2\Delta^2 [ C' +\log^*(n)+\max_i K(a_i)]^2$. Finally let $\zeta$ be a constant less than or equal to $C' + \max_{i \le l} K(a_i)$

Then  for any $ t>\gamma'(n)$,

\begin{equation}\begin{split}&\label{concentration1} P\Big(|\frac{1}{n}K(X_{1:n})-H (X_1|X_{0:-m+1})|>t \Big) \le 2 e^{-\frac{ n (t-\gamma'(n))^2}{ K'_1(n)}}\sim 2e^{-\frac{nt^2}{2\Delta^2\log^*(n)^2}},\end{split}\end{equation}

Furthermore, for any $t>\gamma_n$ we have
\begin{equation}\begin{split}&\label{concentration2}
\mathbb{P}\Big(\Big|\frac{K(X_{1:n})}{n}  -H (X_1|X_{0:-m+1}) \Big|\ge t  \Big) \le  2 e^{-\frac{n (t-\gamma_n)^2}{ K_1}} + n  \zeta 2^{-n^{\frac{1}{2}-\eta}},
\end{split}\end{equation}
. %C is a constant that depends only on the universal machine.

\end{theorem}

%\begin{remark}
%
%Note that the conditions imposed for  Theorem \ref{thm:quantized} are not sufficient for Theorem \ref{thm:expo non-markov}. Indeed
% $\alpha$-mixing does not imply $\phi-$mixing.  Hence having $\alpha(n)=o(n^{-2})$ does not imply that $\sum_k \phi(k) <\infty$.
%However for an $m-$Markov chain on a space with cardinality $l$, if $\rho \triangleq \min_{x_{1:-m} \in A^{m+2}}  P(x_1|x_{0:-m})>0$ and if $\alpha(n)\rightarrow 0$, then  $\phi(n) \le l  (1-\rho^{l^m})^{\lfloor \frac{n}{l^m} \rfloor}$. A proof of this fact can be found in Section \ref{}.
%\end{remark}

Theorem \ref{thm:expo non-markov} can be formulated in the following slightly different way:

\begin{corollary}\label{cor:finitesample}
Let $\{X_i\}_{-\infty}^{\infty}$ be  a  $m-$markov process that satisfies all the conditions of  Theorem \ref{thm:expo non-markov}. Fix $K_1$ to be the value defined in Theorem \ref{thm:expo non-markov}. 
%\begin{itemize}
%\item $\epsilon \in (0, \frac{1}{2}-\frac{1}{2 \kappa})$
%\item Let $m_n \triangleq \frac{\frac{1}{2}-\epsilon}{\log(l)} \log(n)$
%\item  $h_n=n^{\aleph}$, where $\aleph\in (0,\frac{1}{2})$ 
%\end{itemize}
Then, for every $\epsilon>0$, $\exists N$ such that $\forall n \ge N$
 $$\frac{P(\sqrt{n}|\frac{K(X_{1:n})}{n}  -H (X_1|X_{0:-m+1}) |\ge t ) }{ 2 e^{- \frac{t^2}{K_1}}} \le 1+\epsilon$$
\end{corollary}
\begin{proof}
A straightforward application of Theorem \ref{thm:expo non-markov}.
\end{proof}

%\vspace{2mm}
%\begin{remark}
%Note that here we obtain a finite-sample bound on the probability of the event $\sqrt{n}|\frac{K(X_{1:n})}{n}  -H (X_1|X_{-m:0}) |> t$. However the bound obtained in corollary \ref{cor:finitesample} is assymptotically suboptimal since $\lim_{ t \rightarrow 0}  2 e^{-\frac{ t^2}{ K_1}} \rightarrow 2$.
%
%
%
%\end{remark}

\section{Proof}

%\textcolor{red}{Please move it to proof section}
%For Lipchitz functions, of random variable, we can sometimes use transport inequalities to obtain its concentration. It happens that the function of interest here will be Lipschitz, under the appropriate metric, if  $M <\infty$.

% Moreover let $||\cdot||_{TV}$ and $||\cdot||_{op}$ denote the total-variation norm on distributions and the operator norm on matrices.\textcolor{green}{To define I think no?}

\subsection{Background on Kolmogorov complexity}
There are two simple results on the Komogorov complexity that we employ in our proofs. We mention these two as simple lemmas that we can refer to later in the proofs of our main results. For the proof of these results a reader may refer to \cite{cover}, Chapter 14 (Example 14.2.7 and Theorem 14.2.4)

\begin{lemma}\label{ex:integer}
Let $n$ denote an integer number. Then we have the following upper-bound on the Kolmogorov complexity of n:
 $$K(n) \le \log^*(n)+c,$$
where
\vspace{-.2cm}
\begin{eqnarray}
 \log^*(x) = \left\{ \begin{array}{ll}
0 & x \leq 1 \\
 1+ \log^*(\log(x))  & x>1,
\end{array}
\right.
\end{eqnarray}
and where $c$ is a constant that depends only on the universal machine.
\end{lemma}
\noindent It is straightforward to show that $\forall n\ge 1, \log^*(n) < 2 \log(n)+2$.  Another result that will be used about the Komogorov complexity in our paper is the following:

\begin{lemma}\label{Kraft}
 Let $\{0,1\}^\infty \triangleq \bigcup_{i=1}^{\infty} \{0,1\}^i$. If $C_v  \triangleq \{ x \in \{0,1\}^{\infty} \ | \ K(x) < v \}$, then  $|C_v| \le 2^v$.
\end{lemma}

\subsection{Background information on mixing sequences}\label{sec:backgroundinformation}

In our proofs we will also use some well-known results on the central limit theorem for the empirical average of weakly dependent sequences. We summarize these results in this section. 
%\begin{theorem} \cite{speedofconvergence} \label{thm:speed_of_convergenceclt}
%Let $X_1, X_2, \ldots$ denote a strong mixing stationary process with: $$\mathbb{E}(X_1)=0.$$
% Define \begin{itemize}
%\item $\sigma_n^2 \triangleq {\rm var}(\sum_{i=1}^n X_i)$ 
%\item $F_n$ the cdf (cumulative distribution function) of $\frac{\sum_{i=1}^n X_i}{\sigma_n}$.
%\end{itemize}
% Suppose $ \exists (\delta, \beta, K) \in (0,1] \times (1,\infty)  \times \mathbb{R}_+^*$ such that :
%\begin{itemize}
%\item  $\mathbb{E}(|X_1|^{2+\delta})< \infty$.
%\item $\alpha(n) \le K n^{-\beta \frac{(2+ \delta)(1+\delta)}{\delta^2}}$ for every $n$.
%\end{itemize} Then, $\sigma^2 \triangleq \mathbb{E}(X_1^2) +2 \sum_k {\rm cov}(X_1,X_k)< \infty$.

%Moreover,if $\sigma^2>0$, then there exists a constant $A$ such that 
%$$\sup_{t} |F_n(t)-\Phi(t)| \le A n^{-\frac{\delta}{2} \frac{\beta-1}{\beta +1}} \mathbb{E}(|X_1|^{2+\delta}),%$$ 
%where $A$ depends only on $\delta$, $K$, and $\beta$.
%\end{theorem}

%In some of our results our mixing coefficient does not follow the condition $\alpha(n) \leq K n^{-\beta \frac{(2+ \delta)(1+\delta)}{\delta^2}}$.

\begin{theorem} \cite{speed2} \label{thm:speed_of_convergenceclt}
Let $\{X_i\}_{i=1}^{\infty}$ denote a stationary process with $\mathbb{E}(X_1)=0$ and $\mathbb{E}(|X_1|^{2+\delta})< \infty$ for some $\delta \in (0,1]$.
Let $n \in \mathbb{N}$ and define 
 $$\sigma_n^2 \triangleq {\rm var}(\sum_{i=1}^n X_i).$$ 
 Suppose that 
 $\frac{\sigma_n^2}{n} \rightarrow \sigma^2$, where $\sigma^2 \in (0,\infty)$. Let  $F_n$ denote the cdf (cumulative distribution function) of $\frac{\sum_{i=1}^n X_i}{\sigma_n}$. If the $\alpha$-mixing coefficients satisfy
 $$\alpha \triangleq \sum_{i=1}^{\infty} (\alpha(i))^{\frac{\delta}{2+\delta}} < \infty,$$
and there exist $k>1$ and $m$ such that the following conditions hold:
\begin{itemize}
\item[(C.1)] $k  \ge \frac{\log(n)}{2 \log(16)}$,
\item[(C.2)] $k^{\frac{3}{2}} 4^k (\alpha(m+1))^{\frac{1}{2+\delta}} \le 1$,
\item[(C.3)]   $2km+1 < n $,
\end{itemize}
then  there is a constant $C$, that does not depend on the process or n, such that for any $n$ that satisfies $\frac{\sigma_n^2}{n} \ge \frac{1}{4} \sigma^2$,  we have 
\vspace{-.3cm}
\begin{equation*}\begin{split} \sup_{t} |F_n(t)-\Phi(t)| & \le C[x^{2+\delta} \frac{(m+1)^{\delta+1}}{B_n^{\delta}} + x^3 \frac{(m+1)^2}{B_n}  + x^2((m+1)^{\frac{1}{2}}+ \alpha^{\frac{1}{2}})\frac{m+1}{B_n}\\& + x^2 (1+\alpha)^{\frac{1}{2}} B_n (\alpha(m+1))^{\frac{\delta}{2(2+\delta)}}  + x((m+1)^{\frac{1}{2}} + \alpha^{\frac{1}{2}}) (\alpha(m+1))^{\frac{\delta}{2+\delta}}], \end{split}\end{equation*}
where $x \triangleq \frac{2\mathbb{E}(|X_1|^{2+\delta})^{\frac{1}{2+\delta}}}{\sigma}$ and 
 $B_n \triangleq \frac{ 2\sigma_n}{\sigma}$. 
\end{theorem}

In the proof of Theorem \ref{thm:quantized} we will approximate the Kolmogorov complexity using  triangular arrays. We would like to show that  the distribution of $S_{n,n}$, the sum of the first $n$ elements of the $n$-th row of a triangular array,  converges to a normal distribution. To obtain that we will  use the following corollary of Theorem \ref{thm:speed_of_convergenceclt}.

\begin{corollary}\label{thm:speedgeneralise}
Let $\{X_{i}^k\}_{i,k=1}^\infty$  be a double-index process. Furthermore, let $\alpha^k(n)$ denote the $\alpha$-mixing coefficients of $\{X_i^k\}_{i=1}^\infty$. Assume  that $\{X_{i}^k\}_{i=1}^{\infty}$ is a stationary process, with $\mathbb{E}(X_1^k)=0$. Suppose there exists a value of $\delta \in (0,1]$ such  that $\forall \zeta>0, \mathbb{E}(|X^k_1|^{2+\delta})=o(k^{\zeta})$, and that $\mathbb{E}(|X^k_1|^{2+\delta})<\infty$, for all $k$. Let $n \in \mathbb{N}$ and define 
 $\sigma_n^2 \triangleq {\rm var}(\sum_{i=1}^n X^n_i).$
 Suppose that 
 $\frac{\sigma_n^2}{n} \rightarrow \sigma^2$, where $\sigma^2 \in (0,\infty)$. Let  $F_n$ denote the cdf (cumulative distribution function) of $\frac{\sum_{i=1}^n X_i^n}{\sigma_n}$. Suppose that there exist $\epsilon>0$ and $\beta>1$ such that
\begin{itemize} \item $ \epsilon < \min(\frac{\delta}{(\beta+1)(\delta+1)}, 1-\delta \frac{\beta-1}{\beta+1})$ ,
\item
$\forall (n,j),~ \alpha^j(n) \le \min( C' (n-j^{\epsilon})^{-\beta \frac{(2+ \delta)(1+\delta)}{\delta^2}},2), $ where $C'$ is a fixed number.
\end{itemize}
  Then,
\begin{equation*}\begin{split} \sup_{t} |F_n(t)-\Phi(t)| =O_n[\mathbb{E}(|X^n_1|^{2+\delta})n^{-\frac{\delta (\beta-1)}{2(\beta+1)}}] =o_n(n^{-\eta}), \mbox{ for all } \eta< \frac{\delta (\beta-1)}{2(\beta+1)}\end{split}\end{equation*}

\end{corollary}

\begin{proof}
We would like to use Theorem \ref{thm:speed_of_convergenceclt}.
Consider the $n^{\rm th}$ sequence $X^n_1, X^n_2, \ldots $. Note that the $\alpha$-mixing coefficient of this sequence $\alpha^n(k)  \leq \min ((k-n^{\epsilon})^{-\beta \frac{(2+ \delta)(1+\delta)}{\delta^2}}, 1)$. Without loss of generality and for notational simplicity, we assume $C'=1$.

 Hence, it is straightforward to see that $\sum_{k=1}^{\infty} \alpha^n(k) \leq n^\epsilon + d$, where $d$ is a fixed number. To derive this inequality we have used the upper-bound $\alpha^n(k) \le 1$, for $k \le n^{\epsilon}$.  Furthermore, for each $n$ we choose $(m,k)$ in Theorem  \ref{thm:speed_of_convergenceclt} in the following way:
\begin{eqnarray}
 m_n = n^{\frac{\delta}{(\beta+1)(\delta+1)}}, \ \ \ \  {\rm and } \ \ \ \  k_n = \frac{1}{4}\log(n). \nonumber
\end{eqnarray}
It is straightforward to show that, for $n$ sufficiently large, Conditions C.1, C.2, C.3, required in Theorem  \ref{thm:speed_of_convergenceclt}, hold. Furthermore, it is straightforward to check that $\sigma_n^2 \ge \frac{n\sigma^2}{4}$ as required in Theorem \ref{thm:speed_of_convergenceclt}. Hence, we obtain
 \begin{equation*}\begin{split} \sup_{t} |F_n(t)-\Phi(t)| & \le C[x_n^{2+\delta} \frac{(m_m+1)^{\delta+1}}{B_n^{\delta}} + x_n^3 \frac{(m_n+1)^2}{B_n} \\& + x_n^2((m_n+1)^{\frac{1}{2}}+ \alpha^{\frac{1}{2}})\frac{m_n+1}{B_n} + x_n^2 (1+\alpha)^{\frac{1}{2}} B_n (\alpha(m_n+1))^{\frac{\delta}{2(2+\delta)}} \\& + x_n((m_n+1)^{\frac{1}{2}} + \alpha^{\frac{1}{2}}) (\alpha(m_n+1))^{\frac{\delta}{2+\delta}}], \end{split}\end{equation*}
where $x_n \triangleq \frac{2 \mathbb{E}(|X^n_1|^{2 +\delta})^{\frac{1}{2+\delta}}}{\sigma}$. It is straightforward to check that the dominant term is $x_n^{2+\delta}\frac{(m_n+1)^{\delta +1}}{B_n^{\delta}}=O_n[\mathbb{E}(|X^n_1|^{2+\delta})n^{-\frac{\delta (\beta-1)}{2(\beta+1)}}]$. Hence, the proof is complete. 
\end{proof}

%\begin{remark}
%Note that if instead of$(k_n)_n$ is as $k_n \sim n$, then this can be easily generalized to $\sup_{t} |%F_{k_n}^n(t)-\Phi(t)|$, where $F_{k_n}^n$ is the CDF of: $\frac{\sum_{i=1}^{k_n} X_i^n}{\sigma_{n,k_n}}$ %and  $\sigma_{n,k_n}^2 \triangleq {\rm var}(\sum_{i=1}^{k_n} X_i^n)$.
%\end{remark}

The following lemma enables us to connect the correlation of two random variables that are respectively $F_0$ and $F_n$ measurable to the mixing coefficients. 

\begin{lemma} \cite{correlation} \label{lem:correlation}
Let the random variables $\xi$, $\eta$ be measurable with respect to $\mathbb{F}_{-\infty}^t$ and $\mathbb{F}_{t+\tau}^{\infty}$ respectively. Suppose that there is  $\delta>0$, such that
$$\mathbb{E}(|\xi|^{2+\delta})<c_1<\infty \ \ \ {\rm and}  \ \ \ \mathbb{E}(|\eta|^{2+\delta})<c_2<\infty.$$ 
Then, $$|\mathbb{E}(\xi \eta)- \mathbb{E}(\xi)\mathbb{E}(\eta)| \le \alpha(\tau)^{1-\frac{2}{2+ \delta}} (4+ 3(c_1^{\beta} c_2^{1-\beta}+c_1^{1-\beta} c_2^{\beta} )),$$
where $\beta \triangleq \frac{1}{2+\delta}$. 
\end{lemma}

 Most concentration inequalities on random processes assume independence. However here we do not want to make such assumptions. 
% The main tools in those cases are the following
%\begin{itemize}
%\item  Stein Exchangeable pair method, and related ideas, notably developed by Chaterjee  \textcolor{green}{Do I need references?}
%\item Log-sobolev inequalities, and related ideas, when the distribution has a density. \textcolor{green}{Same question}
%\item Martingale methods have been generalized to process verifying some mixing-conditions.
%\end{itemize}
%This list is of course not exhaustive. We are here going to use a result by L.Kontorovich \& K. Ramanan that proposes a method that falls in the last category. Their result is presented in the proof section. 
In the proof of Theorem \ref{thm:expo non-markov} we will use the following result by Kontorovich and Ramanan that  generalizes the martingale method to dependent variables: 
\begin{lemma}\label{lmm: concentration}\cite{martingalemethod}
Suppose that $\Omega$ is a countable space, and let $\{X_i\}_{i=-\infty}^{\infty}$ be a stationary process with $ X_i \in \Omega$. Furthermore, let  $g:\Omega^n \rightarrow \mathbb{R}$ be a 1-Lipschitz function  with respect to the Hamming metric on $\Omega^n$. Define $\Phi^{'}_{i,j} \triangleq \sup_{x_{0:i}, y_{0:i}} \|P(X_{j:n} \in \cdot |X_{0:i}=x_{0:i}) - P(X_{j:n} \in \cdot|X_{0:i}=y_{0:i})\|_{TV}$. Let $H_n$ be an $n \times n$ matrix defined in the following way: 
$$\mbox{~} H_{n,\{i,j\}}= \begin{cases}1 \mbox{, if i=j} \\ \Phi^{'}_{i,j} \mbox{, if i$<j$}\\ 0 \ \ \mbox{\rm otherwise }\end{cases}.$$ 
Then, for all $t>0$ we have
\begin{itemize}
\item $P(g(X_{1:n})- \mathbb{E}(g(X_{1:n})) \ge t) \le  e^{-\frac{t^2}{2n \Delta_n^2}},$
\item $P(g(X_{1:n})- \mathbb{E}(g(X_{1:n})) \le - t) \le  e^{-\frac{t^2}{2n \Delta_n^2}},$
\end{itemize}
where $\Delta_n \triangleq \|H_n\|_{\infty}=\max_{i\le n} (1+ \Phi_{i,i+1}+\dots+\Phi_{i,n})$. %\textcolor{yellow}{A verifier}

\end{lemma}
\begin{remark}
Note that the lemma proposed in \cite{martingalemethod} has a two-sided bound. Here we use a one-sided version. Furthermore note that the conditions on  $\Phi^{'}_{i,j}$ is a bit stronger than the one proposed in \cite{martingalemethod}, but for simplicity we use this condition.
\end{remark}

\subsection{Proof of Lemma \ref{lem:HMMmixing}} \label{sec:lemHMMproof}

According to Proposition 1 in  \cite{HMM1}, there exits $\tau \in (0,1)$ such that
 $$\mathbb{E}(\|P(Y_0\in \cdot |Y_{-1:-m})-P(Y_0 \in \cdot|Y_{-1:-\infty})\|_{TV}) \le C \tau^{m}. $$
Also, note the following three facts that are straightforward to prove:
\begin{enumerate}
\item[(i)] $h(x) \triangleq  \frac{x |\log(x)|^{\frac{2+\delta}{1+\delta}}}{|x-1|}$ is an increasing function of  $ x \in (1,\infty)$.
\item[(ii)]  $h(x) \le 2 $ for $x \in (0,1]$.
\item[(ii)]   $\frac{dP(Y_0|Y_{-1:-\infty})}{dP(Y_0|Y_{-1:-m})} = \frac{\int_{x_0}  dP(x_0|Y_{-1:-\infty}) g(Y_0| x_0)dx_0  }{\int_{x_0}  dP(x_0|Y_{-1:-m}) g(Y_0| x_0)dx_0} \leq \frac{{\rm esssup}_{x_0} g(Y_0 | x_0) }{ {\rm essinf}_{x_0} g(Y_0|x_0) } \leq \eta$. 
\end{enumerate}

By employing these facts we obtain
\begin{equation*}\begin{split}&
\mathbb{E}(|\log(\frac{P(Y_0|Y_{-1:-\infty})}{P(Y_0|Y_{-1:-m})})|^{\frac{2+\delta}{1+\delta}})
\\& =\mathbb{E}(\int |\log(\frac{dP(Y_0|Y_{-1:-\infty})}{dP(Y_0|Y_{-1:-m})})|^{\frac{2+\delta}{1+\delta}} dP(Y_0|Y_{-1:-\infty}))
\\&=\mathbb{E}(\int |\log(\frac{dP(Y_0|Y_{-1:-\infty})}{dP(Y_0|Y_{-1:-m})})|^{\frac{2+\delta}{1+\delta}} \frac{dP(Y_0|Y_{-1:-\infty})}{dP(Y_0|Y_{-1:-m})}d P(Y_0|Y_{-1:-m}))
\\& \le \max(2,\frac{\eta |\log(\eta)|^{\frac{2+\delta}{1+\delta}}}{|\eta-1|}) \mathbb{E}(\int |1- \frac{dP(Y_0|Y_{-1:-\infty})}{dP(Y_0|Y_{-1:-m})}|  d P(Y_0|Y_{-1:-m}))
\\& \le  2  \max(2,\frac{\eta |\log(\eta)|^{\frac{2+\delta}{1+\delta}}}{|\eta-1|}) \mathbb{E}(\|P(Y_0\in \cdot |Y_{-1:-m})-P(Y_0 \in \cdot|Y_{-1:-\infty})\|_{TV})
\\& \le C' \tau^{m}.
\end{split}\end{equation*}
Note that similar ideas have been used in \cite{inequ.kl} .

\subsection{Proof of Remark \ref{upper_bound_M}} \label{sec:proofremarkMstable}
For $n\in \mathbb{N}$, consider the two vectors $x,x' \in A^n$ such that $d_n(x,x')\le 1$. If $d_n(x,x')= 0$, then we can easily see that $|\log(P(x)- \log(P(x'))|=0$. Hence, we assume that $d_n(x,x')= 1$. Suppose that $x_i \neq x'_i$. If $i\in [2,|n-m-1|]$, then
\begin{equation*}\begin{split}&
|\log(P(x_{1:n})- \log(P(x'_{1:n}))| \\&\le |\log(P(x_{1:i-1})- \log(P(x'_{1:i-1}))|+ \sum_{j=i}^{m+1+i} |\log(P(x_j|x_{j-1:j-m}))- \log(P(x'_j|x'_{j-1:j-m}))|\\&~~~ + |\log(P(x_{m+1+i:n}|x_{i+1:m+i}))- \log(P(x'_{m+1+i:n}|x'_{i+1:m+i}))| 
\\& =\sum_{j=i}^{m+1+i} |\log(P(x_j|x_{j-1:j-m}))- \log(P(x'_j|x'_{j-1:j-m}))|  \le -(m+1)\log(\rho).
\end{split}\end{equation*}
This comes from the following facts: (i) For every $j<i$, $x_{1:j}=x'_{1:j}$. Hence, $|\log(P(x_{1:i-1})- \log(P(x'_{1:i-1}))|=0$,  (ii) For every $j>m+1+i$, $x_{j:n}=x'_{j:n}$. Hence, $ |\log(P(x_{m+1+i:n}|x_{i+1:m+i}))- \log(P(x'_{m+1+i:n}|x'_{i+1:m+i}))| =0$.  (iii) Finally, $\forall i \in [|i,m+1+i|]$ $|\log(P(x_j|x_{j-1:j-m}))- \log(P(x'_j|x'_{j-1:j-m}))|\le -\log(\rho)$. The proof for $i \notin [2,|n-m-1|]$ is similar and is hence skipped.

\subsection{Proof of Theorem \ref{thm:quantized}}\label{sec:prooffinitecase}

\subsubsection{Lower bound}
\begin{proof}
Before we discuss the details of the proof, we give a brief overview of the proof strategy to help the reader navigate through the proof more easily. Consider the sequence $X_1, X_2, \ldots, X_n$ with $X_i \in A,$ for all $i$. We assume that $|A| =l$. In this section, we first present a simple program that a universal computer can use to generate this sequence. 

Define $m_n \triangleq \frac{\frac{1}{2} -\epsilon}{\log(l)} \log(n)$, where $\frac{1}{2}-\frac{1}{2C}>\epsilon >0$. Note that $C$ is the same constant as the one used in Condition 3 in the statement of the theorem. The program first tells the universal computer the first $m_n$ bits in the sequence. Then, counts the number of times each $(m_n+1)-$tuple is present in the remaining sequence and reports it.\footnote{For instance, if $m_n=1$, then for the sequence $01001$ the couple $(0,1)$ is present twice, the couple $(1,0)$ once and $(0,0)$ once.} In other words, if we define
\begin{equation}\label{frequence}f_j^{m_n,n} \triangleq \frac{\sum_{k=m_n+1}^{n} \mathbb{I}_{X_{k-m_n:k}=a_j^{m_n}}}{n-m_n},\end{equation}
where $a_j^{m_n}$ is the $j^{\rm th}$ element (in a specific order that is described to the universal computer) of $A^{m_n}$, then the numbers $f_j^{m_n,n}$ are described to the universal computer. Let $\mathbf{f}^{m_n, n}$ denote the vector of all the empirical counts, i.e.,
 \[
 \mathbf{f}^{m_n, n} \triangleq (f_1^{m_n, n}, f_2^{m, n}, \ldots, f_{l^{m_n+1}}^{m_n, n}).
 \] 
 Define an operator $O_f : A^n \rightarrow [0,1]^{l^{m_n+1}}$ that takes $X_1, X_2, \ldots, X_n$ as input and returns $\mathbf{f}^{m_n, n}$ as its output. Then, define the type of a sequence $X_{1:n}$ as the following set:
 \[
 \mathcal{T}_{X_{1:n}} \triangleq \{ Z_{1:n} \ : \   O_f (X_{1:n}) = O_f(Z_{1:n})\  {\rm and} \ Z_{1:m_n} = X_{1:m_n}  \}. 
 \]
Given the information known to the universal computer so far, it has already access to $ \mathcal{T}_{X_{1:n}}$. The only remaining piece of information that the universal computer should have to reconstruct the entire sequence is the index of the sequence $X_1, X_2, \ldots, X_n$ among all the sequences in its type. Let's count the number of bits we have used so far to describe the sequence. 

 Our description requires bits to specify the following quantities: (i)$m_n$, (ii)each $a_j$, (iii) the first $m_n$ bits, (iv)the frequency of observing each possible block of length $(m_n+1)$ in $X_{1:n}$, (v) a systematic way to build all the sequences of length $n$ in $ \mathcal{T}_{X_{1:n}}$, (vi) the index of $X_{1:n}$ in $ \mathcal{T}_{X_{1:n}}$.

\begin{itemize}
\item[(i)] $K(m_n) \le \log^*(m_n)+c$.
\item[(ii)] To describe each $a_j$ at most $l \max_{j \le l} K(a_j)$ are required. 
\item[(iii)] To describe the first $m_n$ symbols we require $m_n( \log^*(l)+c)$.   
\item [(iv)] To describe the frequency of each block we require $l^{m_n+1} \log^*(n)$ bits. The reason is clear, there are $l^{m_n+1}$ different $l$-ary  blocks of length $m_n+1$. Each of them can have at most $n$ elements in them.  
\item[(iv)] So far the universal computer has detected $\mathcal{T}_{X_{1:n}}$.  Now we should describe which element of $\mathcal{T}_{X_{1:n}}$ $X_{1:n}$ is. As the first step we  write a constant size program so that the universal computer realizes what ordering of sequences we are using. The next step is to specify the index of our sequence in this list. To evaluate the number of bits required for describing the index we count the number of elements in $\mathcal{T}_{X_{1:n}} $ . 
\end{itemize}

 Define  $\tilde{P}^{m_n}$ as a new measure on $X_1, X_2, \ldots, X_n$ that has the following properties:
 \begin{enumerate}
 \item  $\tilde{P}^{m_n}$ has the $m_n$-Markov property, i.e.,
 \[
 \tilde{P}^{m_n}(X_1, \ldots, X_n) = \tilde{P}^{m_n}(X_1, X_2, \ldots, X_{m_n}) \prod_{j=m_n+1}^n \tilde{P}^{m_n}(X_j \ | \ X_{j-1}, \ldots, X_{j-m_n}).
 \]
 \item The $m_n+1^{\rm th}$-dimension transition probabilities are the same as those of the original distribution $P$, i.e.,
 \[
 \tilde{P}^{m_n}(X_j \ | \ X_{j-1}, \ldots, X_{j-m_n}) =  {P}(X_j \ | \ X_{j-1}, \ldots, X_{j-m_n}). 
 \]
 \end{enumerate}
 For notational simplicity we consider the notation \begin{equation}\label{Q}Q_j^{m_n} \triangleq \tilde{P}^{m_n}(X_{m_n}=a^{m_n}_{j,m_n}|X_{0}=a^{m_n}_{j,0},\dots, X_{m_n-1}=a^{m_n}_{j,m_n-1}),\end{equation} where $(a_{j,0}^{m_n}, \ldots, a^{m_n}_{j, m_n})$ is the $j^{th}$ element of $A^{m_n+1}$. With this new notation we count the number of elements in $\mathcal{T}_{X_1:n}$. Note that the first $m_n$ symbols are already known. Let's call them
 $x_1, x_2, \ldots, x_{m_n}$. Since,
 \[
 \sum_{X_{1:n} \in\mathcal{T}_{X_{1:n}} }\tilde{P}^{m_n} (X_{m_n+1}, \ldots, X_n \ | \ X_1=x_1, X_2= x_2, \ldots, X_{m_n} = x_{m_n}) \leq 1,
 \]
 we have
 \[
 \sum_{X_{1:n} \in\mathcal{T}_{X_{1:n}} } \prod_{j=m_n+1}^n \tilde{P}^{m_n} (X_{j} |  X_{j-1}, \ldots, X_{j-m_n}) =  |\mathcal{T}_{X_{1:n}}| \prod_{j=1}^{\ell^{m_n+1}} (Q_j^{m_n})^{(n-m_n) f_j^{m_n,n}}
 \]
Hence,
\[
|\mathcal{T}_{X_{1:n}}| < 2^{- (n-m_n) \sum_{j=1}^{l^{(m_n+1)}} f_j^{m_n,n} \log Q_j^{m_n}}. 
\]
This implies that to code the index of an element of $\mathcal{T}_{X_{1:n}}$, we require less than $- (n-m_n) \sum_{j=1}^{l^{m_n+1}} f_j^{m_n,n} \log Q_j^{m_n}$ bits. Combining all the above pieces we obtain the following upper bound for the length of our program: 
\begin{equation} \label{eq:upperKS}
K(X_{1:n}) \le C' + \log^*(m_n) + l \max_{j \le l} K(a_j)  + l^{(m_n+1)} \log ^* n  +m_n \log^*l - (n-m_n) \sum_{j=1}^{l^{m_n+1}} f_j^{m_n,n} \log Q_j^{m_n} \end{equation}
 Our goal is to show that 
$\frac{K(X_1, \ldots, X_n)}{\sqrt{n}} - \sqrt{n}H(X_1 | X_{0:-\infty})$ converges in distribution to a  normal random variable. Note that the first five terms in \eqref{eq:upperKS} are deterministic and when divided by $\sqrt{n}$, they converge to zero.  Hence, we focus on the only remaining term, i.e., $ (n-m_n) \sum_{j=1}^{l^{m_n+1}} f_j^{m,n} \log Q_j^{m_n}$. We have 
\begin{equation}\label{eq:maineq1}
  \begin{split}  &
\sqrt{n}(\frac{1}{n}  (n-m_n) \sum_{j=1}^{l^{m_n+1}} f_j^{m_n,n} \log Q_j^{m_n}+H(X_0|X_{-1},\dots,X_{-\infty}))
\\ &= \sqrt{n}(\frac{n-m_n}{n}  \sum_{j=1}^{l^{m_n+1}} f_j^{m_n,n} \log Q_j^{m_n}+H(X_0|X_{-1},\dots,X_{-m_n})) \\& + \sqrt{n} (H(X_0|X_{-1},\dots,X_{-\infty}) - H(X_0|X_{-1},\dots,X_{-m_n})).  
\end{split}\end{equation}
Our first claim is that 
\begin{equation}\label{eq:convergence_conditional}
\sqrt{n} (H(X_0|X_{-1},\dots,X_{-\infty}) - H(X_0|X_{-1},\dots,X_{-m_n})) \rightarrow 0,
\end{equation}
as $n \rightarrow 0$. To see why this holds, note that
\begin{eqnarray*}
\lefteqn{\sqrt{n} |H(X_0|X_{-1},\dots,X_{-\infty}) - H(X_0|X_{-1},\dots,X_{-m_n})) |} \nonumber \\ 
&\leq& \sqrt{n}  \mathbb{E} | \log P(X_0 \ | \ X_{-1}, \ldots, X_{\infty})  - \log P(X_0 \ | \ X_{-1}, \ldots, X_{-m_n}) | \nonumber \\
&\overset{(a)}{\leq}& \sqrt{n}  (\mathbb{E} | \log P(X_0 \ | \ X_{-1}, \ldots, X_{\infty})  - \log P(X_0 \ | \ X_{-1}, \ldots, X_{-m_n}) |^{\frac{2+\delta}{1+\delta} })^{\frac{1+\delta}{2+\delta}} \nonumber \\
&=& \sqrt{n} (\nu^{\delta}(m_n))^{\frac{1+\delta}{2+\delta}} \rightarrow 0,
\end{eqnarray*}
as $n \rightarrow \infty$. Here we should remind the reader that we have picked $m_n = \frac{\frac{1}{2}-\epsilon}{\log l} \log n$ with $\epsilon$ satisfying  $\frac{1}{2}-\frac{1}{2C}> \epsilon>0$. Note that to obtain (a) we have used Holder inequality and the last step is derived form condition 2 of the theorem regarding the decay of $\nu_\delta$. Combining \eqref{eq:maineq1} and \eqref{eq:convergence_conditional} we conclude that the only remaining step is to show that
 $\sqrt{n}(\frac{n-m_n}{n}  \sum_{j=1}^{l^{m_n+1}} f_j^{m_n,n} \log Q_j^{m_n}+H(X_0|X_{-1},\dots,X_{-m_n}))$ is Gaussian. Toward this goal we first define  
 $$Y_j^{m_n}\triangleq  \log {P}(X_j | X_{j-1}, X_{j-2}, \ldots, X_{j-m_n}).$$
 Note that $\sum_{j=1}^{l^{m_n+1}} f_j^{m_n,n} \log Q_j^{m_n} = \frac{1}{n-m_n}\sum_{j=m_n+1}^n Y_j^{m_n}$. Define
 \[
 S_{n}^{m_n} \triangleq \sum_{i=m_n+1}^n Y_i^{m_n}. 
 \]
To prove the Gaussianity of $S_n^{m_n}$ we employ Corollary \ref{thm:speedgeneralise}. First let us check the conditions of this theorem for $Y_j^{m_n}$: 
 \begin{enumerate}
 \item Boundedness of $\mathbb{E} |Y_j^{m_n}|^{2+\delta} $: First note that
\begin{eqnarray}\label{eq:upperboundingmoment}
\mathbb{E} |Y_j^{m_n}|^{2+\delta} &=& \mathbb{E} \sum_{x_j \in A} {P}(x_j | X_{j-1}, X_{j-2}, \ldots, X_{j-m_n} )| \log {P}(x_j | X_{j-1}, X_{j-2}, \ldots, X_{j-m_n})|^{2+\delta} \nonumber \\
&=&  \mathbb{E} \sum_{x_j \in A} g({P}(x_j | X_{j-1}, X_{j-2}, \ldots, X_{j-m_n} )),
\end{eqnarray}
where the function $g$ is defined in the following way: $g:[0,1] \rightarrow \mathbb{R}$ and $g(t) = t |\log(t)|^{2+\delta}$ for $t \neq 0$, and also $g(0)=0$. It is straightforward to check the following properties of $g$:
\begin{itemize}
\item [(i)] g(t) is continuous at zero.
\item [(ii)] There exists $C_\delta \in (0,1)$ such that $g'(C_\delta) =0$
\item [(iii)]$g'(t) >0$ for $t < C_\delta$ 
\item [(iv)]$g'(t)\leq 0$ for $t> C_\delta$.
\end{itemize} This automatically implies that $g(t) \leq g(C_\delta)$ for all $t \in [0,1]$. Combing this fact with \eqref{eq:upperboundingmoment} implies
\begin{equation}\label{eq:boundedness}
\mathbb{E} |Y_j^{m_n}|^{2+\delta} = \mathbb{E} \sum_{X_j \in A} g({P}(X_j | X_{j-1}, X_{j-2}, \ldots, X_{j-m_n} ) \leq l g(C_\delta). 
\end{equation}
Note that the upper bound does not depend on either $m_n$, $n$ or $j$. 
 \item The mixing coefficient $\alpha$: First let $\alpha^{Y^{m_n}}(i)$ denote the $\alpha$-mixing coefficient for the $Y^{m_n}$ sequence, and let $\alpha(i)$ denote the $\alpha$ mixing coefficient for the original process $X_1, \ldots, X_n$. It is straightforward to check that for every $i> m_n$
 \[
 \alpha^{Y^{m_n}} (i) \leq \alpha(i-m_n) \leq \begin{cases} K (i-m_n)^{-\beta \frac{(2+\delta) (1+\delta)}{\delta^2}} , ~i>m_n\\ 1 \mbox{, otherwise.}\end{cases}
 \] 
 where the last step is due to Condition 2 in the statement of the theorem. As a reminder we have $m_n =O(\log(n))$.
 \item 
For notational simplicity in the rest of the proof we use the notation $\sum_{j=1}^{n-m_n} Y_j^{m_n}$ instead of $\sum_{j=m_n+1}^n Y_j^{m_n}$. Define $\tilde{\sigma}_n^2 =  \rm{var}(Y_{1:n-m_n}^{m_n})$. We will later prove that $\frac{\tilde{\sigma}_n^2}{n} \rightarrow \sigma^2$, where 
{\small \[
\sigma^2 \triangleq {\rm var} (\log(P(X_0|X_{-1},\dots,X_{-\infty})))  + 2 \sum_{k} {\rm cov}(\log(P(X_0|X_{-1},\dots,X_{-\infty})),\log(P(X_k|X_{k-1},\dots,X_{-\infty}))).
\] }
%This will be done from \eqref{eq:combinedvariance} onward, but before we get to the proof of this result we prove that $0<\sigma^2<\infty$. We first prove that $\sigma^2< \infty$.% Similar to the proof we presented in \eqref{eq:upperboundingmoment} one can prove that $ {\rm var} (\log(P(X_0|X_{-1},\dots,X_{-\infty}))) < \infty$. By definition we have
%\begin{eqnarray}
%\sum_{k} {\rm cov}(\log(P(X_0|X_{-1},\dots,X_{-\infty})),\log(P(X_k|X_{k-1},\dots,X_{-\infty})))
%=  \sum_k {\rm cov} (Y^{m_n}_{-\infty:0}, Y^{m_n}_{-\infty:k}). 
%\end{eqnarray}
First we can see that $\sigma^2<\infty$. In that goal define  
\[
 W_j \triangleq \log(P(X_j|X_{j/2:j-1})). 
 \]
 We have
{\small \begin{equation*}\begin{split}\label{eq:eqlastref9}&
\lefteqn{\sum_k {\rm cov} (\log(P(X_0|X_{-1},\dots,X_{-\infty})), \log(P(X_k|X_{k-1},\dots,X_{-\infty})))}\\& = \sum_k  {\rm cov} (\log(P(X_0|X_{-1},\dots,X_{-\infty})), W_k)+ \sum_k {\rm cov} (\log(P(X_0|X_{-1},\dots,X_{-\infty})), \log(P(X_k|X_{k-1},\dots,X_{-\infty}))-W_k)\nonumber
 \\&\overset{(a)}{\leq} \sum_k \alpha(\frac{k}{2})^{\frac{\delta}{\delta+2}} (4+ 6l g(C_\delta)) + \sum_k (\nu_\delta(k/2))^{\frac{1+\delta}{2+\delta}} \nonumber 
 \\& \le K  (4+ 6l g(C_\delta)) \sum_k  n^{-\frac{\beta (1+\delta)}{\delta}}  + \sum_k 2^{-C\log(\ell)\frac{k}{2}} < \infty.\nonumber 
\end{split}\end{equation*}}
 To obtain the first term in Inequality (a) we employed Lemma \ref{lem:correlation}. To obtain the second term after Inequality (a) we used Holder's inequality and Definition \ref{def:nudelta}. The last inequality is the result of Condition 2 in the statement of our theorem. 

 \end{enumerate}

 We can now  prove that $\frac{\tilde{\sigma}_n^2}{n} \rightarrow \sigma^2$.  We have 
\begin{eqnarray} \label{eq:combinedvariance}
\frac{{\rm var}(\sum_{j=1}^{n-m_n} Y_j^{m_n})}{n-m_n} &=& {\rm var}(Y_1^{m_n}) + \frac{2}{n} \sum_{i=1}^n \sum_{k = i+1}^n {\rm cov}(Y_i^{m_n},Y_k^{m_n}) \nonumber \\
& =& {\rm var}(Y_1^{m_n}) + \frac{2}{n} \sum_{i=1}^n \sum_{k= 2}^{i} {\rm cov}(Y_1^{m_n},Y_{k}^{m_n}),
 \end{eqnarray}
 where to obtain the last equality we used the stationarity of the process $Y_1^{m_n}, Y_2^{m_n}, \ldots$. 
Our goal is to show that this quantity converges to $\sigma^2$. We simplify the expression of \eqref{eq:combinedvariance} in the following two steps:

\begin{enumerate}
\item Simplifying ${\rm var}(Y_1^{m_n})$:  First note that 
\begin{eqnarray}
\lefteqn{|\mathbb{E}(\log(P(X_1|X_{0:-m_n+1}))) - \mathbb{E}(\log(P(X_1|X_{0:-\infty})))| }  \nonumber \\
\leq&& \!\!\!\!\!\!\!\!\!\! (\mathbb{E} |(\log(P(X_1|X_{0:-m_n+1}))) - \log(P(X_1|X_{0:-\infty}))|^{\frac{2+\delta}{1+\delta}} )^{\frac{1+\delta}{2+\delta}} = (\nu_\delta(m_n))^{\frac{1+\delta}{2+\delta}} \rightarrow 0.
\end{eqnarray}
To obtain the last inequality we used Holder's and to obtain the last convergence we used Condition 2 in the statement of the theorem. Furthermore, note that
\begin{eqnarray}
\lefteqn{|\mathbb{E}(\log^2(P(X_1|X_{0:-m_n}))) - \mathbb{E}(\log^2(P(X_1|X_{0:-\infty})))| }  \nonumber \\
&\leq& (\mathbb{E} |(\log(P(X_1|X_{0:-m_n+1}))) - \log(P(X_1|X_{0:-\infty}))|^{\frac{2+\delta}{1+\delta}} )^{\frac{1+\delta}{2+\delta}} \nonumber \\
&&\times (\mathbb{E} |(\log(P(X_1|X_{0:-m_n+1}))) + \log(P(X_1|X_{0:-\infty}))|^{{2+\delta}} )^{\frac{1}{2+\delta}} \rightarrow 0.
\end{eqnarray}
To prove the last convergence we should note that the first term goes to zero according to Condition 2 in the statement of the theorem. Furthermore, similar to the proof of \eqref{eq:boundedness} we can show that the last expectation is bounded. Hence, it is straightforward to combine the above two equations and obtain 
\begin{equation}\label{eq:varianceconvergence1}
{\rm var}(Y_1^{m_n}) ={\rm var}(\log(P(X_1|X_{0:-m_n+1}))) \rightarrow {\rm var}(\log(P(X_1|X_{0:-\infty}))). 
\end{equation}
\item Our second step is to discuss the covariance terms in \eqref{eq:combinedvariance}. Define
 \begin{eqnarray}\label{eq:differences_sn}
  s_{i,n} &\triangleq& \sum_{k=2}^{i}{\rm cov}(Y_1^{m_n},Y_k^{m_n}), \nonumber \\
  s &\triangleq & \sum_j {\rm cov} (\log(P(X_1|X_{-1:-\infty})),\log(P(X_j|X_{j-1:-\infty}))))\nonumber.
  \end{eqnarray}
  Note that our goal is to bound
  \begin{eqnarray}
  \frac{1}{n} |\sum_{i=1}^n (s_{i,n}-s)| \leq \frac{1}{n} \sum_{i=1}^{2 m_n} |s_{i,n}-s| + \frac{1}{n} \sum_{i=2 m_n+1}^{n} |s_{i,n}-s|. 
  \end{eqnarray}
  We will prove later that $\sup_i |s_{i,n}-s|$ is bounded. Hence, since $m_n/n \rightarrow 0$, we conclude that the first term goes to zero.  Hence, we focus on the second term. Define $Z_j \triangleq \log(P(X_j|X_{j-1:-\infty}))$. Then we have 
  \begin{eqnarray}\label{eq:maintermsn_sdiff}
\frac{1}{n} \sum_{i=2m_n+1}^{n} |s_{i,n}-s| &\leq&  \frac{1}{n}  \sum_{i=2m_n+1}^{n} \sum_{j=2}^{2m_n} |{\rm cov}(Y_1^{m_n},Y_j^{m_n}) - {\rm cov} (Z_1,Z_j) | \nonumber \\
&+& \frac{1}{n}  \sum_{i=2m_n+1}^{n} \sum_{j=2 m_n+1}^{i} |{\rm cov}(Y_1^{m_n},Y_j^{m_n}) - {\rm cov} (Z_1,Z_j) | \nonumber \\
&+& \frac{1}{n}  \sum_{i=2m_n+1}^{n} \sum_{j=i}^{\infty} |{\rm cov} (Z_1,Z_j) |. 
  \end{eqnarray}
  We will show that the each of the three terms on the right converge to zero. Before we proceed further, note that
  \begin{eqnarray}
  \mathbb{E} (|Y_1^{m_n}-Z_1|^{\frac{2+\delta}{1+\delta}}) = \mathbb{E} |\log(P(X_1|X_{0:-m_n+1})) -\log(P(X_1|X_{0:-\infty}))|^{\frac{2+\delta}{1+\delta}} = \nu_\delta(m_n).
\end{eqnarray}
Furthermore, similar to the proof of \eqref{eq:boundedness} it is straightforward to show that  
\begin{eqnarray}
\mathbb{E} (|Z_j|) \leq (\mathbb{E} |Z_j|^{2+\delta})^{\frac{1}{2+\delta}} <M, \nonumber \\
\mathbb{E} |Y_j^{m_n}| \leq (\mathbb{E} |Y_j^{m_n}|^{2+\delta})^{\frac{1}{2+\delta}}< M, 
\end{eqnarray}
where $M^{2+\delta} = l  \sup_{t \in [0,1]} |g_{2+\delta}(t)|$ with $g_{2+\delta}(t) = t |\log(t)|^{2+\delta}$. Now we turn our attention to bounding the terms in \eqref{eq:maintermsn_sdiff}.  
  \begin{eqnarray}\label{eq:boundingcovariances}
 \lefteqn{ |{\rm cov}(Y_1^{m_n},Y_j^{m_n}) - {\rm cov} (Z_1,Z_j) | 
 \leq |{\rm cov}(Y_1^{m_n}-Z_1,Z_j) |+| {\rm cov} (Y_1^{m_n},Z_j-Y_j^{m_n}) |} \nonumber 
 \\&\leq& \mathbb{E} |(Y_1^{m_n}-Z_1)Z_j| + | \mathbb{E}(Y_1^{m_n}-Z_1) \mathbb{E}Z_j| +  \mathbb{E} |Y_1^{m_n}(Z_j-Y_j^{m_n})| + | \mathbb{E}(Y_1^{m_n}) \mathbb{E}(Z_j-Y_j^{m_n})| \nonumber \\
  &\leq& (\mathbb{E} |Y_1^{m_n}-Z_1|^{\frac{2+\delta}{1+\delta}})^{\frac{1+\delta}{2+\delta}}(\mathbb{E} |Z_j|^{2+\delta})^{\frac{1}{2+\delta}}  + (\mathbb{E} |Y_j^{m_n}-Z_j|^{\frac{2+\delta}{1+\delta}})^{\frac{1+\delta}{2+\delta}} \mathbb{E} |Z_j| \nonumber \\
  &&+  (\mathbb{E} |Y_j^{m_n}-Z_j|^{\frac{2+\delta}{1+\delta}})^{\frac{1+\delta}{2+\delta}}(\mathbb{E} |Y_1^{m_n}|^{2+\delta})^{\frac{1}{2+\delta}}  + (\mathbb{E} |Y_j^{m_n}-Z_j|^{\frac{2+\delta}{1+\delta}})^{\frac{1+\delta}{2+\delta}} \mathbb{E} |Y_1^{m_n}| \nonumber \\
&\leq& 4M (\nu_\delta(m_n))^{\frac{1+\delta}{2+\delta}}.  
 \end{eqnarray} 
Hence, we conclude that
\[
 \frac{1}{n}  \sum_{i=2m_n+1}^{n} \sum_{j=1}^{2m_n} |{\rm cov}(Y_1^{m_n},Y_j^{m_n}) - {\rm cov} (Z_1,Z_j) | \leq \frac{n-2m_n}{n} 2 m_n 4M (\nu_\delta(m_n))^{\frac{1+\delta}{2+\delta}} \rightarrow 0,
\]
 as $n \rightarrow \infty$. Note that the last convergence in the theorem is derived from Condition 2 in the statement of the theorem.  Now we find a bound on the second term in \eqref{eq:maintermsn_sdiff}. Define 
 \[
 W_j \triangleq \log(P(X_j|X_{j/2:j-1})). 
 \]
 Then, we have
   \begin{eqnarray}\label{eq:boundingcovariances}
 \lefteqn{ |{\rm cov}(Y_1^{m_n},Y_j^{m_n}) - {\rm cov} (Z_1,Z_j) | 
 \leq |{\rm cov}(Y_1^{m_n}-Z_1,Z_j) |+| {\rm cov} (Y_1^{m_n},Z_j-Y_j^{m_n}) |} \nonumber \\
 &\leq&   |{\rm cov}(Y_1^{m_n}-Z_1,Z_j -W_j) |+ |{\rm cov}(Y_1^{m_n}-Z_1,W_j)| + | {\rm cov} (Y_1^{m_n},Z_j-Y_j^{m_n}) | \nonumber \\
 &\leq&   |{\rm cov}(Y_1^{m_n}-Z_1,Z_j -W_j) |+ |{\rm cov}(Y_1^{m_n}-Z_1,W_j)| + | {\rm cov} (Y_1^{m_n},Z_j-W_j) | \nonumber \\
 &&+  | {\rm cov} (Y_1^{m_n},W_j) | +  | {\rm cov} (Y_1^{m_n},Y_j^{m_n}) |
 \end{eqnarray}
 The strategy that we use to bound the terms $|{\rm cov}(Y_1^{m_n}-Z_1,Z_j -W_j) |$ and $| {\rm cov} (Y_1^{m_n},Z_j-W_j) |$ is the same. Also, the strategy we use to bound $ |{\rm cov}(Y_1^{m_n}-Z_1,W_j)|$ and $| {\rm cov} (Y_1^{m_n},W_j) |$ is the same. Hence, we only derive the bounds for the following three terms: (i) $| {\rm cov} (Y_1^{m_n},Z_j-W_j) |$ , (ii) $ | {\rm cov} (Y_1^{m_n},W_j) |$, and (iii) $ | {\rm cov} (Y_1^{m_n},Y_j^{m_n}) |$. 
\begin{enumerate}
\item $| {\rm cov} (Y_1^{m_n},Z_j-W_j)|$: By using holder inequality we conclude that 
\begin{eqnarray}
|{\rm cov} (Y_1^{m_n},Z_j-W_j)|  &\leq& \mathbb{E} |Y_1^{m_n}(Z_j-W_j)| + \mathbb{E} |Y_1^{m_n}| \mathbb{E}|Z_j-W_j|   \nonumber \\
&\leq& 2(\mathbb{E} |Z_j-W_j|^{\frac{2+\delta}{1+\delta}})^{\frac{1+\delta}{2+\delta}} \mathbb{E} (|Y_1^{m_n}|^{2+\delta})^{\frac{1}{2+\delta}} \nonumber \\
&\leq& 2(\nu_\delta(j/2))^{\frac{1+\delta}{2+\delta}}  M. 
\end{eqnarray}

\item $| {\rm cov} (Y_1^{m_n},W_j) |$: Note that $W_j$ is measurable with respect to $\mathcal{F}_{j/2}^j$ and $Y_1^{m_n}$ is measurable with respect to $\mathcal{F}^{1}_{-\infty}$. Hence, by employing Lemma 7 we conclude that 
\[
| {\rm cov} (Y_1^{m_n},W_j) | \leq \alpha(j/2)^{\frac{\delta}{\delta+2}} (4+2\tilde{M}),
\]
where $\tilde{M} = l g(C_\delta)$. Note that to obtain the last inequality we have used \eqref{eq:boundedness}. 

\item  $ | {\rm cov} (Y_1^{m_n},Y_j^{m_n}) |$: Similar to the argument of the previous case we conclude that 
\[
 | {\rm cov} (Y_1^{m_n},Y_j^{m_n}) |  | \leq \alpha(j-m_n)^{\frac{\delta}{\delta+2}} (4+2\tilde{M}).
\]
\end{enumerate}

Combining \eqref{eq:boundingcovariances} and the above three cases, we conclude that 
\begin{eqnarray}
\lefteqn{\frac{1}{n}  \sum_{i=2m_n+1}^{n} \sum_{j=2 m_n+1}^{i} |{\rm cov}(Y_1^{m_n},Y_j^{m_n}) - {\rm cov} (Z_1,Z_j) |} \nonumber \\
\leq&& \! \!\!\!\!\!\!\!\! \frac{1}{n} \!\!\! \sum_{i=2m_n+1}^{n} \sum_{j=2 m_n+1}^{\infty} 4 (\nu_\delta(j/2))^{\frac{1+\delta}{2+\delta}}  M+ 2 \alpha(j/2)^{\frac{\delta}{\delta+2}} (4+2M) +  \alpha(j-m_n)^{\frac{\delta}{\delta+2}} (4+2M) \nonumber \\
\leq && \! \!\!\!\!\!\!\!\!\!\!\! \sum_{j=2 m_n+1}^{\infty} 4 (\nu_\delta(j/2))^{\frac{1+\delta}{2+\delta}}  M+ 2 \alpha(j/2)^{\frac{\delta}{\delta+2}} (4+2M) +  \alpha(j-m_n)^{\frac{\delta}{\delta+2}} (4+2M) \rightarrow 0,
\end{eqnarray}
as $n \rightarrow \infty$. 
The last term of \eqref{eq:maintermsn_sdiff} can be bounded in exactly similar fashion, i.e., we use the upper bound $|{\rm cov}(Z_1, Z_j) | \leq |{\rm cov}(Z_1, Z_j- W_j)| +   |{\rm cov}(Z_1,W_j)| $, and then employ Lemma 7 and the definition of $\nu_\delta$ to bound the error. Since the proof is similar we skip it. 
\end{enumerate}
Combining all these steps we conclude that 
\begin{equation}\label{eq:covarianceconvergence}
 \frac{1}{n} |\sum_{i=1}^n (s_{i,n}-s)| \rightarrow 0. 
\end{equation}
 Equations \eqref{eq:combinedvariance}, \eqref{eq:varianceconvergence1}, and \eqref{eq:covarianceconvergence} together prove that
 \[
 \frac{{\rm var}(\sum_{j=1}^{n-m_n} Y_j^{m_n})}{n-m_n} \rightarrow \sigma^2. 
 \]

Therefore if $\sigma^2=0$ we have proved that:$$\frac{1}{\sqrt{n-m_n}}\sum_{j=1}^{n-m_n} [Y_j^{m_n}-H(X_{1}|X_{0:-m_n+1})]\xrightarrow{L_2}0. $$ Hence \begin{equation} 
\frac{1}{\sqrt{n}}[ C' + \log^*(m_n) + l \max_{j \le l} K(a_j)  + l^{(m_n+1)} \log ^* n  - m_n \log^*l - (n-m_n) \sum_{j=1}^{l^{m_n+1}} f_j^{m_n,n} \log Q_j^{m_n}]\xrightarrow{d}0. \end{equation}

 Now  if $\sigma^2>0$ we can apply Corollary \ref{thm:speedgeneralise}, with $F_k^{m_n}$ denoting the CDF of $\frac{\sum_{j=m_n+1}^{m_n+k} Y_j^{m_n}- k H(X_1|X_0,\dots,X_{-m_n}) }{\sqrt{{\rm var}(\sum_{j=m_n+1}^{m_n+k} Y_j^{m_n})}}$. By employing the triangle inequality we have
\begin{equation}\label{eq:maineq2}
  \begin{split}  &
\sup_t \left|P \left(\frac{\sqrt{n-m_n}}{\sigma}(\sum_{j=1}^{l^{(m_n+1)}} f_j^{m_n,n} \log Q_j^{m_n}+H(X_0|X_{-1},\dots,X_{-m_n})) \le t \right) -\Phi(t)\right|
\\&   \le \sup_t \Big|\Phi(t)- \Phi( t\frac{  \sqrt{n-m_n}\sigma}{\sqrt{{\rm var}(\sum_{j=m_n+1}^{n} Y_j^{m_n})}}) \Big| +\sup_t |F_{n-m_n}^{m_n}(t) -\Phi(t)|.
\end{split}\end{equation}

According to Corollary \ref{thm:speedgeneralise}, $\sup_t |F_{n-m_n}^{m_n}(t) -\Phi(t)|=o((n-m_n)^{-\frac{\delta(\beta-1)}{4(\beta+1)}})=o(n^{-\frac{\delta(\beta-1)}{4(\beta+1)}}).$% Hence, $\sup_t \Big|\Phi(t)- \Phi( t\frac{  \sqrt{n-m_n}\sigma}{\sqrt{{\rm var}(\sum_{j=m_n+1}^{n} Y_j^{m_n})}}) \Big|  \rightarrow 0$. 
Moreover we have proved that,
$$\frac{\sqrt{n-m_n}\sigma}{\sqrt{{\rm var}(\sum_{j=1}^{n-m_n} Y_j^{m_n})}} \rightarrow 1.$$
By employing the mean value theorem we can then show that: 
\[
 \sup_t \Big|\Phi(t)- \Phi( t\frac{  \sqrt{n-m_n}\sigma}{\sqrt{{\rm var}(\sum_{j=1}^{n-m_n} Y_j^{m_n})}}) \Big|  \rightarrow 0,
\]
as $n \rightarrow \infty$. If we use this in \eqref{eq:maineq2} we conclude that 
$\lim\inf_{n \rightarrow \infty} P(\sqrt{n} (\frac{1}{n} K(X_{1:n} )- H(X_0|X_{-1},\dots,X_{-\infty}))\le t) \ge \Phi(t \sigma)$, which is one side of what we had to prove. 

\end{proof}
%%%%%%%%%%%%%%%%%%%%%%%%%%%%%%%%%%%%%%%%%%%%%
\subsubsection{Upper bound}\label{ssec:upperbound}
\begin{proof}
Define $\delta_n \triangleq n^{-\frac{2}{3}}$.

\begin{eqnarray}\label{eq:entropyconnection}
\lefteqn{ P\Big(\frac{K(X_{1:n})}{n} < -\frac{\log(P(X_{1:n}))}{n} - \delta_n  \Big) } \nonumber \\
 & \le&\!\! P\left(\frac{K(X_{1:n})}{n} < -\frac{\log(P(X_{1:n}))}{n} - \delta_n , \frac{K(X_{1:n})}{n} <x \right) + P\left(\frac{K(X_{1:n})}{n}>x\right).  \end{eqnarray}
 Our goal is to show that under a proper choice of $x$, both probabilities on the right converge to zero as $n \rightarrow \infty$. First note that
 \begin{eqnarray}\label{eq:upperkfinitecasefirst}  
 \lefteqn{P\left(\frac{K(X_{1:n})}{n} < -\frac{\log(P(X_{1:n}))}{n} - \delta_n , \frac{K(X_{1:n})}{n} <x \right) } \nonumber \\ & \le& \sum_{i=1}^{nx} \sum_{\substack{v\mbox{ as }K(v)=i, \\ P(v) <2^{-(i+ n \delta_n)}} } P(X_{1:n}=v) \nonumber  \le \sum_{i=1}^{nx} \sum_{\substack{v\mbox{ as }K(v)=i, \nonumber \\ P(v) <2^{-(i+ n \delta_n)}} } 2^{\log P(v)} \nonumber \\& \le&  \sum_{i=1}^{nx} \sum_{\substack{v\mbox{ as }K(v)=i,  \\ P(v) <2^{-(i+ n \delta_n)}} } 2^{-(i+ n \delta_n)}  \le \sum_{i=1}^{nx} 2^i 2^{-(i+ n \delta_n)}  \le nx 2^{-n \delta_n } \rightarrow 0. \end{eqnarray}

Furthermore, if we choose $x= \frac{3}{2} H(X_0|X_{-1},\dots, X_{-\infty})$, we have 
\begin{equation}\label{eq:upperboundfirst3}
 P\left(\frac{K(X_{1:n})}{n}>\frac{3}{2} H(X_0|X_{-1},\dots, X_{-\infty})\right)  \rightarrow 0.
\end{equation}
as $n \rightarrow \infty$. Hence, by combining \eqref{eq:entropyconnection}, \eqref{eq:upperkfinitecasefirst}, and \eqref{eq:upperboundfirst3}, we have
\begin{equation}\label{eq:Kolmoboundfinal}
P\Big(\frac{K(X_{1:n})}{n} < -\frac{\log(P(X_{1:n}))}{n} - \delta_n  \Big)  \rightarrow 0.
\end{equation}
On the other hand, $\forall t$, 
{\small \begin{eqnarray*}  
\lefteqn{ P\Big(\sqrt{n}(\frac{1}{n} K(X_{1:n} )- H(X_0|X_{-1},\dots,X_{-\infty}))\le t \Big) } \nonumber \\  \le&& \!\!\!\!\!\!\!\! P\Big(\frac{K(X_{1:n})}{n} < -\frac{\log(P(X_{1:n}))}{n} - \delta_n  \Big) + P\Big(\sqrt{n}( -\frac{\log(P(X_{1:n}))}{n} - \delta_n -  H(X_0|X_{-1},\dots,X_{-\infty})) \le t \Big). 
\end{eqnarray*} }
Note two main points about our last expression: (i) According to \eqref{eq:Kolmoboundfinal} the first term goes to zero as $n \rightarrow \infty$. (ii) We would like to characterize the limiting distribution of  $( -\frac{\log(P(X_{1:n}))}{\sqrt{n}} - \sqrt{n}\delta_n -  \sqrt{n}H(X_0|X_{-1},\dots,X_{-\infty})) $. We rewrite this expression in the following way:
\begin{eqnarray}
\lefteqn{-\frac{\log(P(X_{1:n}))}{\sqrt{n}} - \sqrt{n}\delta_n -  \sqrt{n}H(X_0|X_{-1},\dots,X_{-\infty})) } 
\nonumber \\
&=&  -\frac{\log(P(X_{1:n}))}{\sqrt{n}} - \sqrt{n}\delta_n + \frac{ \sum_{j=m_n+1}^{n} \log(P(X_j|X_{j-1:j-m_n}))}{\sqrt{n}} \nonumber \\
&&- \frac{ \sum_{j=m_n+1}^{n} \log(P(X_j|X_{j-1:j-m_n}))}{\sqrt{n}}+  \sqrt{n}H(X_0|X_{-1},\dots,X_{-\infty})).
\end{eqnarray}
where $m_n= \frac{\frac{1}{2} -\epsilon}{\log l} \log n$, where $\frac{1}{{2}}-\frac{1}{2C}> \epsilon>0$. Note that if we prove
\begin{equation}\label{eq:almostdonefirsttheoremfinite1}
-\frac{\log(P(X_{1:n}))}{\sqrt{n}} - \sqrt{n}\delta_n + \frac{ \sum_{j=m_n+1}^{n} \log(P(X_j|X_{j-1:j-m_n}))}{\sqrt{n}} \overset{p}{\rightarrow} 0,
\end{equation}
and
\begin{equation}\label{eq:almostdonefirsttheoremfinite2}
- \frac{ \sum_{j=m_n+1}^{n} \log(P(X_j|X_{j-1:j-m_n}))}{\sqrt{n}}+  \sqrt{n}H(X_0|X_{-1},\dots,X_{-\infty}))  \overset{d}{\rightarrow} N(0, \sigma^2),
\end{equation}
then by Slutsky's theorem we conclude that 
\[
 P\Big(\sqrt{n}( -\frac{\log(P(X_{1:n}))}{n} - \delta_n -  H(X_0|X_{-1},\dots,X_{-\infty})) \le t \Big) \rightarrow \Phi(\sigma t). 
\]
Proof of \eqref{eq:almostdonefirsttheoremfinite2} is the same as the proof we presented in the last section. To prove \eqref{eq:almostdonefirsttheoremfinite1} first note that $\sqrt{n} \delta_n \rightarrow 0$. Furthermore, 
\begin{eqnarray}
\lefteqn{\mathbb{E} \left( \left|\frac{\log(P(X_{1:n}))}{\sqrt{n}}-\frac{ \sum_{j=2m_n}^{n} \log(P(X_j|X_{j-1:j-m_n-1}))}{\sqrt{n}} \right|  \right) } \nonumber \\
&\leq & \mathbb{E} \left( \left|\frac{\log(P(X_{1:m_n}))}{\sqrt{n}}+ \frac{ \sum_{j=m_n+1}^{n} \log(P(X_j|X_{1:j}))}{\sqrt{n}}-\frac{ \sum_{j=m_n+1}^{n} \log(P(X_j|X_{j-1:j-m_n}))}{\sqrt{n}} \right|  \right) \nonumber \\
& \le& -\mathbb{E}\left(\frac{\log(P(X_{1:m_n}))}{\sqrt{n}}\right) + \frac{1}{\sqrt{n}} \sum_{j=m_n+1}^{n} \mathbb{E} (|\log(P(X_j|X_{1:j-1})) -\log(P(X_j|X_{j-m_n:j-1})) |) \nonumber \\
& \le& -\mathbb{E}\left(\frac{\log(P(X_{1:m_n}))}{\sqrt{n}}\right) + \frac{1}{\sqrt{n}} \sum_{j=m_n+1}^{n} \mathbb{E} (|\log(P(X_j|X_{1:j-1})) -\log(P(X_j|X_{-\infty:j-1})) |) \nonumber \\
&&+\frac{1}{\sqrt{n}} \sum_{j=m_n+1}^{n} \mathbb{E} (|\log(P(X_j|X_{-\infty:j-1})) -\log(P(X_j|X_{j-m_n:j-1})) |) \nonumber \\
&\leq&  -\mathbb{E}\left(\frac{\log(P(X_{1:m_n}))}{\sqrt{n}}\right)  + \frac{1}{\sqrt{n}}  \sum_{j=m_n+1}^{n} (\nu_\delta(j))^{\frac{1+\delta}{2+\delta}} + \frac{n-m_n}{\sqrt{n}}  \left(\nu_\delta(m_n) \right)^{\frac{1+\delta}{2+\delta}} \rightarrow 0,
\end{eqnarray}
as $n\rightarrow \infty$.  Hence, $\forall t$ $\lim \sup_{n \rightarrow \infty} P(\sqrt{n} (\frac{1}{n} K(X_{1:n} )- H(X_0|X_{-1},\dots,X_{-\infty}))\le t) \le \Phi(t \sigma)$.

\end{proof}

\subsection{Proof of Theorem \ref{thm:expo non-markov}}

Before we go to the details of the proof we will review the main ideas. We are going to use the  upper and lower bounds on the Kolmogorov complexity, derived in the proof  of Theorem \ref{thm:quantized} to get inequality \ref{concentration2}. For each bound we will obtain concentration-inequalities and combine them to obtain a concentration result for the Kolmogorov complexity. We use the concentration inequality presented in Lemma \ref{lmm: concentration}. Note that we use the notations defined in  \eqref{frequence} and \eqref{Q}. Define 
$$ g(X_{1:n}) \triangleq (n-m) \sum_{j=1}^{l^{m+1}} f_j^{m,n} \log Q_j^{m}.$$ 
We would like to use Lemma \ref{lmm: concentration} to show that $g(X_{1:n})$ concentrates. Toward this goal we need to do the following two steps: (i) Calculate an upper bound for $\Delta_n=\|H_n\|_{\infty}$, where $H_n$ is the $n \times n$ matrix with elements $$\forall (i,j), \mbox{~} H_{n,\{i,j\}}= \begin{cases}1 \mbox{, if i=j} \\ \Phi^{'}_{i,j} \mbox{, if i$<j$}\\ 0 \mbox{ otherwise }\end{cases}.$$ (ii)   $g(X_{1:n})$ is a 1-Lipschitz for the Hamming-distance.%, which is true almost-surely.

To show inequality \ref{concentration1} we would also to use Lemma \ref{lmm: concentration}. Toward this goal we also need  to do the following two steps: (i) Prove that  $X_{1:n}\rightarrow\frac{K(X_{1:n})}{n}$ is a Lipschitz function for the Hamming-distance. (ii) Calculate an upper-bound for $\mathbb{E}(\frac{K(X_{1:n})}{n})- H(X_{1}|X_{0:-m+1})$. With this summary we now discuss the details of the proof.

First, we bound $\Delta_n$. For every $(j,n)$ define
$$A_{j,n}(x_{0:i}) \triangleq \{x_{j:n} \in A^{n-j+1} \mbox{, such that } P(X_{j:n}=x_{j:n} |X_{0:i}=x_{0:i})- P(X_{j:n}=x_{j:n} ) \ge 0\} .$$ Then, 
\begin{equation*}\begin{split} 
&  \sup_{x_{0:i}} \|P(X_{j:n}\in \cdot|X_{0:i}=x_{0:i})- P(X_{j:n}\in \cdot) \|_{TV}
\\& \le \sup_{x_{0:i}} \sum_{x_{j:n} \in A^{n-j}} |P(X_{j:n}=x_{j:n} |X_{0:i}=x_{0:i})- P(X_{j:n}=x_{j:n} ) |
\\&  \le \sup_{x_{0:i}}[ \sum_{x_{j:n} \in A_{n,j}(x_{0:i})} P(X_{j:n}=x_{j:n} |X_{0:i}=x_{0:i})- P(X_{j:n}=x_{j:n} )
\\& \mbox{~~~~~~~~} +   \sum_{x_{j:n} \in  {A^c_{n,j}(x_{0:i})}}P(X_{j:n}=x_{j:n} )-  P(X_{j:n}=x_{j:n} |X_{0:i}=x_{0:i})]
\\& \le \sup_{x_{0:i}}[ P(X_{j:n}\in A_{n,j}(x_{0:i})|X_{0:i}=x_{0:i})- P(X_{j:n}\in A_{n,j}(x_{0:i}))
\\& \mbox{~~~~~~~~} - P(X_{j:n}\in A^c_{n,j}(x_{0:i})|X_{0:i}=x_{0:i})+ P(X_{j:n}\in A_{n,j}^c(x_{0:i}))]
\\& \le 2  \sup_{A \in \mathbb{F}_{-\infty}^{i}, B \in \mathbb{F}_{j}^{\infty}} |P(B|A)-P(A)| \le 2 \phi(j-i).
\end{split}\end{equation*}
Hence, according to the definition of the $\Phi^{'}_{i,j}$ we have
\begin{equation*}\begin{split} \Phi^{'}_{i,j}& =\sup_{x_{0:i}, y_{0:i}} \|P(X_{j:n}\in \cdot|X_{0:i}=x_{0:i})- P(X_{j:n}\in \cdot|X_{0:i}=y_{0:i}) \|_{TV}
\\& \le  2 \sup_{x_{0:i}} \|P(X_{j:n}\in \cdot|X_{0:i}=x_{0:i})- P(X_{j:n}\in \cdot) \|_{TV}
\\& \le 4 \sup_{A \in \mathbb{F}_{-\infty}^{i}, B \in \mathbb{F}_{j}^{\infty}} |P(B|A)-P(A)| \le 4 \phi(j-i).
\end{split}\end{equation*}

And so we have that $\Delta_n \le 1+ 4 \sum_{k=0}^{\infty} \phi(k)< \infty$. Moreover, according to the proof of Theorem  \ref{thm:quantized}, using the notations introduced in \eqref{eq:upperKS}, we have
{\small $$ K(X_{1:n}) \le C' + \log^*(m) + l \max_{j \le l} K(a_j)  + l^{(m+1)} \log ^* n  - m \log^*l - (n-m) \sum_{j=1}^{l^{m+1}} f_j^{m,n} \log Q_j^{m}.$$}

Let $C_1(n) \triangleq C' + \log^*(m) + l \max_{j \le l} K(a_j)  + l^{(m+1)} \log ^* n  - m \log^*l$. Our goal it to find a concentration inequality  for $ (n-m) \sum_{j=1}^{l^{m+1}} f_j^{m,n} \log Q_j^{m}$. Toward this goal, we prove that this function is 1-Lipschitz and then use Lemma \ref{lmm: concentration}. Note that
$$ (n-m) \sum_{j=1}^{l^{m+1}} f_j^{m,n} \log Q_j^{m} =\sum_{j=m+1}^n \sum_{k=1}^{l^{m+1}} I_{(X_{j-m:j}= a_k^{m})} \log(Q_k^{m}),$$
where $a_k^m$ is the $k^{\rm th}$ element of $A^m$. Let $x, x' \in A^n$ denote two vectors that only differ at  the $j^{\rm th}$-coordinate (i.e. $x_{i}=x'_i, ~\forall i\ne j$).  Then, by the $M$-stability assumption of the theorem $|g(x)-g(x')| \le M$ (note that $g$ is the log-likelihood of $X_{m+1:n}$). Hence, $g$ is $M$-Lipschitz for the Hamming metric. Lemma \ref{lmm: concentration} implies that for every $t>0$  
$$ \mbox{~}P\Big((n-m) \sum_{j=1}^{l^{m+1}} f_j^{m,n} \log Q_j^{m}+ (n-m) H(X_1|X_{0:-m+1}) \le - t\Big) \le 2 e^{-\frac{t^2}{ 2n M^2 \Delta^2}},$$
and
$$\mbox{~}P\Big((n-m) \sum_{j=1}^{l^{m+1}} f_j^{m,n} \log Q_j^{m}+ (n-m) H(X_1|X_{0:-m+1}) \ge  t\Big) \le 2 e^{-\frac{t^2}{ 2n M^2 \Delta^2}}.$$
It is straightforward to confirm that for every $t$, if $ t  -\frac{C_1(n)}{n}+  \frac{m}{n} H(X_{1}|X_{0:-m+1}) >0$, then
{\footnotesize \begin{equation} \label{eq:concentrate2}\begin{split}&
P\Big(\frac{ K(X_{1:n})}{n} - H(X_{1}|X_{0:-m+1}) \ge t \Big) 
\\& \le  P\Big(\frac{1}{n} ((n-m)
  \sum_{j=1}^{l^{m+1}} f_j^{m,n} \log Q_j^{m}- [(n-m) H(X_1|X_{0:-m+1}) +m H(X_1|X_{0:-m+1})]) \ge t -\frac{C_1(n)}{n}\Big) 
\\& \le P\Big(\frac{1}{n} ((n-m)
  \sum_{j=1}^{l^{m+1}} f_j^{m,n} \log Q_j^{m}- (n-m) H(X_1|X_{0:-m+1})) \ge t  -\frac{C_1(n)}{n}+  \frac{m}{n} H(X_{1}|X_{0:-m+1}) \Big) 
\\& \le e^{-\frac{n \Big( t  -\frac{C_1(n)}{n}+  \frac{m}{n} H(X_{1}|X_{0:-m+1}) \Big)^2}{ 2 M^2 \Delta^2}}.
\end{split}\end{equation}}

%In the second inequality we can note that we used the well known following fact: 

%$\forall m \in \mathbb{N}, \mbox{~} H(X_1|X_{0:-m})  \ge H(X_1|X_{0:-\infty}).$

To prove the upper bound, first set $\delta_n= \frac{1}{n^{\frac{1}{2}+\eta}}$. Similar to the proof we presented in Section \ref{ssec:upperbound}, we can prove that
$$P\Big(\frac{K(X_{1:n})}{n} < -\frac{\log(P(X_{1:n}))}{n} - \delta_n  \Big)  \le n \zeta e^{-n^{\frac{1}{2}-\eta}}.$$
 Hence, if $n$ satisfies $ t + \frac{m }{n} H(X_1|X_{0:-m+1})>0$, then
{\footnotesize \begin{equation}\begin{split}& \label{lowerbound}
P\left(\frac{K(X_{1:n})}{n}  -H (X_1|X_{0:-m+1})  \le - t \right)
\\&  \le P\Big( -\frac{\log(P(X_{1:n}))}{n} - \frac{n-m}{n}H (X_1|X_{0:-m+1})  \le  -t + \delta_n + \frac{m}{n}  H(X_1|X_{0:-m+1}) \Big) \\& \mbox{~~~}+ P\Big(\frac{K(X_{1:n})}{n} < -\frac{\log(P(X_{1:n}))}{n} - \delta_n  \Big) 
\\& \le P\Big( -\frac{1}{n} [(n-m)
  \sum_{j=1}^{l^{m+1}} f_j^{m,n} \log Q_j^{m}- (n-m) H(X_1|X_{0:-m+1})]  \le - t + \frac{m }{n} H(X_1|X_{0:-m+1}) \Big)\\& \mbox{~~}  + P\Big(\frac{K(X_{1:n})}{n} < -\frac{\log(P(X_{1:n}))}{n} - \delta_n  \Big) 
\\& \le e^{-\frac{n \left( t - \frac{m }{n} H(X_1|X_{0:-m+1})-\delta_n\right)^2}{ 2 M^2 \Delta^2}} + n \zeta  e^{-n^{\frac{1}{2}-\eta}}.
\end{split}\end{equation}}
To obtain the second inequality we used the fact that $-\log(P(X_{1:n})) =- \log(P(X_{1:m}))- \log(P(X_{m+1:n}|X_{1:m})) \ge - \log(P(X_{m+1:n}|X_{1:m})) $. Finally,
The first term in the last line is similar to \eqref{eq:concentrate2}. Hence, by combining \eqref{eq:concentrate2} and \eqref{lowerbound} we obtain 
$$P \left(\left|\frac{K(X_{1:n})}{n}  -H (X_1|X_{0:-m+1}) \right|\ge t \right) \le 2e^{-\frac{n \left(t  -\frac{C_1(n)}{n}-  \frac{m}{n} H(X_{1}|X_{0:-m+1})-\delta_n \right)^2}{ 2 M^2 \Delta^2}} + n\zeta e^{-n^{\frac{1}{2}-\eta}}.$$

Finally, note that $C_1(n)=O_n(n^{-1} \log^*(n))$, and hence if we define $\gamma_n \triangleq  \frac{C_1(n)}{n}+ \frac{m}{n} H(X_{1}|X_{0:-m+1})+\delta_n$ and $K_1 \triangleq 2 M^2 \Delta^2 $, 
then we have 
$$P\left( \left|\frac{K(X_{1:n})}{n}  -H (X_1|X_{-m:0}) \right|\ge t \right) \le 2 e^{-\frac{n (t -\gamma_n)^2}{ K_1}} + n\zeta 2^{-n^{\frac{1}{2}-\eta}},$$
where $\gamma_n= O(n^{-(\frac{1}{2}-\eta)})$.

We now want to discuss the details of the proof of inequality \ref{concentration1}.

For $n\in \mathbb{N}$, consider the two vectors $x,x' \in {A^n}^2$ such that $d_n(x,x')\le 1$. If $d_n(x,x')= 0$, then we can easily see that $|K(x_{1:n})- K(x'_{1:n})|=0$. Hence, we assume that $d_n(x,x')= 1$. Suppose that $x_i \neq x'_i$.  Then  we can note that if the universal machine knows $x_{1:n}$ to know $x'_{1:n}$ it only need to know i and $x'_i$. Therefore $$K(x'_{1:n})\le   K(x_{1:n})+ C' +\log^*(n)+\max_i K(a_i),$$ where $C'$ is a constant that depends only on the universal machine.

As the previous inequality is symetric in $x,x'$ we obtain that  $x_{1:n}\rightarrow\frac{1}{n}K(x_{1:n})$ is $\frac{ C' +\log^*(n)+\max_i K(a_i)}{n}$-Lipschitz.

Lemma \ref{lmm: concentration} implies that for every $t>0$  
$$P\Big(|\frac{1}{n}K(X_{1:n})-\mathbb{E}(\frac{1}{n}K(X_{1:n}))|>t \Big) \le 2 e^{-\frac{ n t^2}{ 2 (C' +\log^*(n)+\max_i K(a_i))^2 \Delta^2}}.$$

Moreover thanks to Kraft inequality and  the positivity of the Kullback-Leiller divergence we have that $\mathbb{E}(\log(\frac{P(X_{1:n})}{2^{-K(X_{1:n})}}))\ge 0$, hence $H(X_{1:n})\ge \mathbb{E}(K(X_{1:n}))$.

Moreover we can use the upper-bound  on the Kolmogorov complexity obtained in equation \ref{eq:upperKS} to get that for all $m\in \mathbb{N}$
\begin{equation*}
\mathbb{E}(K(X_{1:n})) \le C' + \log^*(m) + l \max_{j \le l} K(a_j)  + l^{m+1} \log ^* n  + m \log^*l + (n-m) H(X_1|X_{0:-m+1}).\end{equation*}

Hence 
\begin{equation*}
|\frac{1}{n}\mathbb{E}(K(X_{1:n}))-H(X_1|X_{0:-m+1})| \le \frac{C' + \log^*(m) + l \max_{j \le l} K(a_j)  + l^{m+1} \log ^* n  + m \log^*l +m H(X_1)}{n}.\end{equation*}

Therefore by defining $\gamma'(n)\triangleq \frac{C' + \log^*(m) + l \max_{j \le l} K(a_j)  + l^{m+1} \log ^* n  + m \log^*l +m H(X_1)}{n}$ we get that $\forall t>\gamma'(n)$

\begin{equation*}
P\Big(|\frac{1}{n}K(X_{1:n})-\mathbb{E}(\frac{1}{n}K(X_{1:n}))|>t \Big) \le 2 e^{-\frac{ n (t-\gamma'(n))^2}{ 2 (C' +\log^*(n)+\max_i K(a_i))^2 \Delta^2}}.\end{equation*}

\subsubsection{Proof of Example \ref{ex:slower}}\label{sec:proofexampleslow}

We first mention the following central-limit theorem for triangular arrays of martingales that will be later used in the proof. 
\begin{theorem}\label{Martingale}\cite{martingaleclt}
Let $(S_{n,i},F_{i}, 1\le i\le k_n, n\ge 1)$ be a  zero-mean, square integrable martingale array with differences $X_{n,i}$, and let $\eta^2$ be an a.s. finite random variable. Suppose that
\begin{equation*}\begin{split}
\forall \epsilon>0,~&\sum_{i\le k_n}\mathbb{E}(X_{n,i}^2 I_{|X_{n,i}|>\epsilon}|F_{i-1})\xrightarrow{P}0
\\& \sum_{i\le k_n}\mathbb{E}(X_{n,i}^2|F_{i-1})\xrightarrow{P}\eta^2.
\end{split}\end{equation*}
Then $S_{n,k_n}\xrightarrow{d}Z$, where the characteristic function $Z$ is $\mathbb{E}(e^{-\frac{1}{2}\eta^2t^2})$.
\end{theorem}

We review the roadmap of the proof. First we find an upper bound and lower-bound for the complexity of $X_{1:n}$ in terms of the $(\tau_k)_k$ and $(Y_k)_k$. Using this upper and lower bound we will prove that there is a function, $f_n$ such that $ \sqrt{n}(\frac{K(X_{1:n})}{n}-f_n(\{\tau_k\}_{k=-\infty}^{\infty}, \{Y_k\}_{k=-\infty}^{\infty})) \rightarrow 0$ almost surely. This implies that if the central-limit theorem holds, then the asymptotic distribution of  $ \sqrt{n}(H(X_1|X_{0:-\infty})-f_n(\{\tau_k\}_{k=-\infty}^{\infty}, \{Y_k \}_{k=-\infty}^{\infty})) $ would also be Gaussian. We will then prove that this does not happen since there is a $\eta>0$ such as: $ n^{\frac{1}{2}- \eta}(H(X_1|X_{0:-\infty})-f_n(\{\tau_k\}_{k=-\infty}^{\infty}, \{Y_k\}_{k=-\infty}^{\infty})) $ is not bounded in probability.

First to understand the proof we have to notice that the process $\{X_i\}_{i=-\infty}^{\infty}$ is constituted of different segments of random variables that comes from different distributions and those segments have different lengths, for example $X_{1:\tau_1-\theta}|\tau_1,\theta$ comes from a certain distribution and $X_{\tau_1-\theta+1 : \tau_1-\theta+\tau_2}|\tau_1,\theta,\tau_2$ may come from another distribution. Let $\{L_i\}_i$ denote the $i^{\rm th}$ segments, e.g. $L_1 = X_{1:\tau_1-\theta}$. Define $l_1 \triangleq \tau_1-\theta$, which is the length of the first segment, and for every $i>0$ define
$$N_i \triangleq \max \{k : \mbox{ such that } l_1 +\tau_2+ \cdots +\tau_k \le i).$$ 
$N_i$ is maximum number of segments $\{L_k\}_k$, including the first one, that are entirely in $X_{1:i}$. Finally, define $l_{left}(i) \triangleq i- l_1- \sum_k^{N_i} \tau_k,$ which is the number of elements of $X_{1:i}$ that are not in any of the different $L_k$, for $k \le N_i$. 

%We can by abuse of notation write: $X_{1:i}=(L_1,\dots, L_{k},X_{ l_1+ \sum_k^{N_i} \tau_k+1:i}).$

To describe $X_{1:n}$ we may describe each segment $X_{1:l_1}$, $X_{l_1+1:l_1+ \tau_2}$,\dots,$X_{l_1+ \sum_k^{N_n} \tau_k +1:n }$. It is straightforward to confirm the following two facts: (i) if $Y_i=1 $ then the $i^{\rm th}$ segment can be described by the length of the segment and a constant cost, C,  to indicate to the machine that it should produce an array of 0's. (ii) If $Y_i=0 $ then the $i^{\rm th}$ segment can be described by describing each element in that segment. Since we have  $N_n+1$ segments, it is straightforward to confirm that
\begin{eqnarray}
\lefteqn{K(X_{1:n} \ | \ N_n, l_1, \tau_2, \ldots, \tau_{N_n}, l_{left(n)} , Y_1, \ldots, Y_{N_n})} \nonumber \\
& \le& C(N_n+1) + \min(l_1,n) I_{Y_1=0} + \sum_{i \le N_n} \tau_i  I_{Y_i=0} + l_{left}(n) I_{Y_{N_n+1}=0}. 
\end{eqnarray}
Note that for the full-description of $X_{1:n}$ we should also describe the following to the universal machine: (i) $(l_1,\tau_{1:N_n},l_{left}(n))$, (ii)  $(Y_{1:N_n+1})$. Hence it is straightforward to check the following upper bound for the Kolmogorov complexity of $X_{1:n}$:
\begin{eqnarray}\label{equ:ubexample2}
K(X_{1:n}) \!\!\!\! & \le& \!\!\! (N_n+1) (1+C+ \log^*(n)) + \min(l_1,n) I_{Y_1=0} + \sum_{i \le N_n} \tau_i  I_{Y_i=0} + l_{left}(n) I_{Y_{N_n+1}=0}.  \nonumber\\
\end{eqnarray}

Before we proceed to simplify the above upper bound, let me find a lower bound for the Kolmogorov's complexity of $X_{1:n}$ as well. Define the vector $V_n$ in the following way: take all the segments of $X_{1:n-l_{left}(n)}$ that are coming from $\tilde{Y}$ and concatenate them to obtain the vector $V_n$. Note that if the Universal computer has access to $X_{1:n}$, then it only requires the following information to construct $V_n$: the values of $Y_{1}, \ldots, Y_{N_n}$ and $l_1, \tau_2, \tau_3, \ldots, \tau_{N_n}$. Hence, it is straightforward to show that
\begin{eqnarray}\label{eq:lbexample2}
K(V_n) \leq K(X_{1:n}) + (N_n +1) (1+\log^*n+C). 
\end{eqnarray}
It is intuitively clear that since $V_n$ has iid $Bern (1/2)$ elements its Kolmogorov complexity should be concentrated around its length. Below we prove this intuition:

\begin{lemma}\label{lem:concentration}
Let $l_n$ denote the length of $V_n$. If $\delta_n = n^{-2/3}$, then 
\[
\mathbb{P} (K(V_n) \leq l_n - n \delta_n|l_n ) \rightarrow 0,
\]
as $n \rightarrow \infty$.
\end{lemma}
\begin{proof}
First for a certain $l_n$ we can describe $V_n$ by :
\begin{itemize}
\item[(i)] Describe the length of the sequence: $l_n$, with a cost of at most $\log^{*}(l_n)+C$.
\item[(ii)] Describe each of the $l_n$ elements of the sequence, with a cost of at most $l_n$.
\item[(iii)] Telling it how to build the sequence, with a cost of $C'$, where $C'$ is a constant that depends only on the universal machine.
\end{itemize} 
Hence:
$$K(V_n) \le C'+\log^{*}(l_n)+ l_n.$$
And so: $P(K(V_n) \ge 2l_n|l_n) \rightarrow 0.$ Then, we have
\begin{equation*} \begin{split}&
P(K(V_n) \le -\log(P(V_n|l_n))- n\delta_n |l_n) 
\\& \le P(K(V_n) \ge 2l_n|l_n) + P(K(V_n) \le 2l_n, K(V_n) \le -\log(P(V_n|l_n))- n\delta_n |l_n) 
\\& \le  P(K(V_n) \ge 2l_n|l_n) + \sum_{i=1}^{2 l_n} \sum_{ v \mbox{ as } K(v) =i} 2^{\log(P(V|l_n))} \le P(K(V_n) \ge 2l_n|l_n) + 2 l_n 2^{n\delta_n} \rightarrow 0.
\end{split}\end{equation*}
Please note that to pass from the second-line to the third we have used Lemma \ref{Kraft}. Finally,
\[
\mathbb{P} (K(V_n) \leq l_n - n \delta_n | l_n) \rightarrow 0.
\]
Indeed knowing $l_n$, $V_n$ is a sequence of iid bernouilli$(\frac{1}{2})$ and so: 
$-\log(P(V_n)|l_n)=l_n $
\end{proof}

We should note that:  $l_n= (l_1\wedge n) I_{Y_1=0} + \sum_{i \le N_n} \tau_i  I_{Y_i=0} + l_{left}(n) I_{Y_{N_n+1}=0}$.

Combing \eqref{equ:ubexample2}, \eqref{eq:lbexample2}, and Lemma \ref{lem:concentration} we obtain the following upper and lower bounds for $K(X_{1:n})$:
\begin{eqnarray}\label{eq:upperlower}
K(X_{1:n}) &\leq& (N_n+1) (1+C+ \log^*(n)) + l_n. \nonumber \\
K(X_{1:n}) &\geq& l_n- (N_n+1) (1+C+ \log^*(n)) - n\delta_n, \nonumber \\
\end{eqnarray}
where the lower bound holds with probability converging to $1$. Our next goal is to show that with probability converging to one
\begin{eqnarray}\label{eq:Nngoeszero}
\frac{1}{\sqrt{n}} \left( K(X_{1:n}) -  \min(l_1,n) I_{Y_1=0} - \sum_{i \le N_n} \tau_i  I_{Y_i=0} - l_{left}(n) I_{Y_{N_n+1}=0} \right) \rightarrow 0. 
\end{eqnarray}
It is straightforward to confirm that $\frac{n \delta_n}{\sqrt{n}}  \rightarrow 0$. Hence, we only have to prove that $N_n (\log^*(n)+C+1)/\sqrt{n} \rightarrow 0$ (which is going to be true if $N_n \log^*(n)/\sqrt{n} \rightarrow 0$). Toward this goal define $S_n \triangleq \sum_{i=2}^n \tau_i$. Since $S_n$ is a sum of iid variables, it is straightforward to confirm that 
\[
\frac{S_n}{S_{n}^{1+u}} \xrightarrow{a.s} 0,
\]
$0<u< \frac{\frac{1}{2}- \frac{\epsilon}{4} }{\frac{1}{2}-\epsilon} -1.$  Moreover as $\frac{S_n}{S_{n}^{1+u}}=\frac{1}{S_n^u}\le \frac{1}{\tau_1^u}\in L^1$, by dominated convergence theorem we also obtain the $L^1$ convergence. Then we have that by exchangeability of the $(\tau_i)_{i\le n}|S_n$ that
\[
\mathbb{E}(\frac{S_n}{S_{n}^{1+u}})  =\mathbb{E}( \mathbb{E}(\frac{\sum_{i=1}^n\tau_i}{S_{n}^{1+u}}|S_n) )= \mathbb{E}(\frac{n \tau_1}{S_{n}^{1+u}}) \xrightarrow{P} 0. 
\]

Therefore $\frac{n \tau_1}{(S_{n}-\tau_1)^{1+u}}\xrightarrow{L_1} 0$, which implies that \begin{equation*}\begin{split}
 \mathbb{E}(\frac{n \tau_1}{(S_{n}-\tau_1)^{1+u}})&\ge
nP(\frac{ \tau_1}{(S_{n}-\tau_1)^{1+u}}\ge 1) \\&=n\mathbb{E} \big(P(   \tau_1 > (S_{n}-\tau_1)^{1+u} | S_{n}-\tau_1) \big)
\\& \ge K'n \mathbb{E}( S_{n}^{-(1+u)(\frac{1}{2}+\epsilon)})\rightarrow0,
\end{split}\end{equation*} 
where we have used the fact that 
there is a constant $K'$ such that any fixed $b$, $P(|\tau_1|>b) \geq K' b^{-(\frac{1}{2}-\epsilon)} $

Hence, 
\begin{equation}\label{eq:Sngrand}
\mathbb{E} ( S_{n}^{-(\frac{1}{2}-\epsilon)(1+u)}) = o(\frac{1}{n}). 
\end{equation}
By employing Markov inequality we obtain $S_{n}^{-1} = o_p(n^{-(1+u)^{-1}(\frac{1}{2}-\epsilon)^{-1}})$. Note that if we have $m \triangleq \lfloor n^{\frac{1}{2}-\frac{\epsilon}{4}}\rfloor.$
\begin{equation*}\begin{split}&
\mathbb{P} (N_n > n^{1/2- \epsilon/4}) \le \mathbb{P} (S_{\lfloor n^{\frac{1}{2}-\frac{\epsilon}{4}}\rfloor} \le n )
\\& \le \mathbb{P} (S_{m} \le m^{(\frac{1}{2}-\frac{\epsilon}{4})^{-1}})  \le \mathbb{P} (1 \le S_{m}^{-1} (m^{(\frac{1}{2}-\frac{\epsilon}{4})^{-1}})) \rightarrow 0 
\end{split}\end{equation*}
Where the last equation comes from Equation \ref{eq:Sngrand} and $(\frac{1}{2}-\frac{\epsilon}{4})^{-1} < \frac{1}{(\frac{1}{2}-\epsilon)(1+u)}$. Hence, it is straightforward to conclude that 
\begin{equation}\label{eq:convergenceNn}
\frac{N_n \log^{*}(n)}{\sqrt{n}} \rightarrow 0. 
\end{equation}

This completes the proof of \eqref{eq:Nngoeszero}.

It is straightforward to prove that the entropy rate of this process is $1/2$. Hence, we would like to show that 
\[
\sqrt{n} \left(\frac{K(X_1, X_2, \ldots, X_n)}{n} - \frac{1}{2}\right), 
\]
is $\omega(1)$. Suppose that this is not the case, then by using Prohorov's theorem the sequence is tight and the sequence $\sqrt{n} (\frac{K(X_1, X_2, \ldots, X_n)}{n} - \frac{1}{2})$ will have a subsequence that converges almost surely. To simplify the notation, instead of working with the convergent subsequence we assume that the entire sequence converges in distribution. Since
 \begin{eqnarray}
\lefteqn{\sqrt{n} \left( \frac{\min(l_1,n) I_{Y_1=0} - \sum_{i \le N_n} \tau_i  I_{Y_i=0} - l_{left}(n) I_{Y_{N_n+1}=0}}{n} -\frac{1}{2} \right)} \nonumber \\
& =& \sqrt{n} \left(\frac{K(X_1, X_2, \ldots, X_n)}{n} - \frac{1}{2} \right) \nonumber \\
&&+  \frac{1}{\sqrt{n}} \left( K(X_{1:n}) -  \min(l_1,n) I_{Y_1=0} - \sum_{i \le N_n} \tau_i  I_{Y_i=0} - l_{left}(n) I_{Y_{N_n+1}=0} \right)
\end{eqnarray}
and according to \eqref{eq:Nngoeszero}: $$ \frac{1}{\sqrt{n}} \left( K(X_{1:n}) -  \min(l_1,n) I_{Y_1=0} - \sum_{i \le N_n} \tau_i  I_{Y_i=0} - l_{left}(n) I_{Y_{N_n+1}=0} \right) \xrightarrow{P} 0,$$
  and we have assumed that  $\sqrt{n} (\frac{K(X_1, X_2, \ldots, X_n)}{n} - \frac{1}{2})$ converges in distribution, we can use Slutsky's theorem and claim that $\sqrt{n} \left( \frac{\min(l_1,n) I_{Y_1=0} - \sum_{i \le N_n} \tau_i  I_{Y_i=0} + l_{left}(n) I_{Y_{N_n+1}=0}}{n} -\frac{1}{2} \right)$ converges in distribution. Note that $l_1  I_{Y_1=0} < \tau_1$ and $l_{left}(n)I_{Y_{N_n+1}=0} < \tau_{N_n+1}$. Therefore, 
\[
  \frac{\min (l_1,n)I_{Y_1=0} + l_{left}(n) I_{Y_{N_n+1}=0}  }{\sqrt{n}}  \xrightarrow{a.s.} 0.
\]
Hence
\[
\sqrt{n} \left( \frac{\min(l_1,n) I_{Y_1=0} + l_{left}(n) I_{Y_{N_n+1}=0}}{n} -\frac{1}{2} \right) \rightarrow 0,
\]
and our analyses reduces to the analysis of $\sqrt{n} \left( \frac{ \sum_{i \le N_n} \tau_i  I_{Y_i=0} }{n} -\frac{1}{2} \right)$. 
Note that
 \begin{eqnarray}
\lefteqn{\sqrt{n} \left(\mathbb{E}  \left( \frac{ \sum_{i \le N_n} \tau_i  I_{Y_i=0} }{n}  | \ N_n, \tau_1, \tau_2, \ldots, \tau_{N_n}\right) -0.5 \right)} \nonumber \\
&=& \sqrt{n} \left( \mathbb{E}  \left( \frac{ \sum_{i \le N_n} \tau_i  I_{Y_i=0} }{n} \ | \ N_n, \tau_1, \tau_2, \ldots, \tau_{N_n}\right)-0.5 \right) \nonumber \\
&=&\sqrt{n} \left( \frac{1}{2}  \frac{\sum_{i\leq N_n} \tau_i }{n}-\frac{1}{2} \right) \nonumber = \frac{\sqrt{n}}{2}  \frac{\sum_{i\leq N_n} \tau_i -n}{n}  =  \frac{l_1 +l_{left}(n)}{2\sqrt{n}} \xrightarrow{a.s.} 0\nonumber. 
%\\
%&=&-  \frac{\min(l_1,n) I_{Y_1=0}+ l_{left}(n) I_{Y_{N_n+1}=0}}{2\sqrt{n}} \rightarrow 0, \nonumber
\end{eqnarray}
 Hence we discuss the limiting distribution of the following quantity:
\[
\sqrt{n} \left( \frac{ \sum_{i \le N_n} \tau_i  I_{Y_i=0} }{n} - \mathbb{E} \left(\frac{ \sum_{i \le N_n} \tau_i  I_{Y_i=0} }{n}  | \ N_n, \tau_1, \tau_2, \ldots, \tau_{N_n} \right)\right).
\]

Toward that goal we will first  introduce the following sigma-fields: 
$$F_l\triangleq \sigma(\tau_i,I_{i\le N_l} I_{Y_i=0},~i \in \mathbb{N}),$$ and the processes  $$Y_l^n\triangleq \frac{1}{\sqrt{\sum_i \tau_i^2 I_{i \le N_n}} } \sum_{i} \tau_i I_{i \le N_l}   (I_{Y_i=0}-\frac{1}{2}).$$

 It is straightforward to see that $((Y^n_l,F_l)_l)_n$ is a triangular array of  martingales. The corresponding martingale differences are given by 
 $$X_{n,i}\triangleq \frac{1}{\sqrt{\sum_i \tau_i^2 I_{i \le N_n}} }\sum_{j=1}^{\infty} \tau_j I_{N_{i-1}<j \le N_i}   (I_{Y_i=0}-\frac{1}{2}). $$ 
We would now like to use Theorem \ref{Martingale}. It is straightforward to check that
$$\frac{1}{\sum_i \tau_i^2 I_{i \le N_n} } \sum_{i=1}^{n} \mathbb{E}( (\sum_j \tau_j I_{N_{i-1}<j \le N_i}   (I_{Y_i=0}-\frac{1}{2}) )^2 |F_{i-1} )=\frac{1}{4}.$$

Furthermore, we have to prove the following claim:
\begin{equation}\label{eq:condthm11}
\forall \epsilon>0,~\sum_{i\le n} \mathbb{E}(\big|\frac{\sum_j \tau_j I_{N_{i-1}<j \le N_i}   (I_{Y_i=0}-\frac{1}{2}) }{\sqrt{\sum_i \tau_i^2 I_{i \le N_n}}} \big|^2I_{|\frac{\sum_j \tau_j I_{N_{i-1}<j \le N_i}   (I_{Y_i=0}-\frac{1}{2}) }{\sqrt{\sum_i \tau_i^2 I_{i \le N_n}}}|>\epsilon}|F_{i-1})\xrightarrow{P} 0. 
\end{equation}
Toward this goal, note that 
\begin{equation*} \begin{split}&
\sum_{i\le n} \mathbb{E} \left(\Big|\frac{\sum_j \tau_j I_{N_{i-1}<j \le N_i}   (I_{Y_i=0}-\frac{1}{2}) }{\sqrt{\sum_i \tau_i^2 I_{i \le N_n}}} \Big|^2I_{\Big|\frac{\sum_j \tau_j I_{N_{i-1}<j \le N_i}   (I_{Y_i=0}-\frac{1}{2}) }{\sqrt{\sum_i \tau_i^2 I_{i \le N_n}}}\Big|>\epsilon}|F_{i-1}\right)
\\& \overset{(a)}{=}\sum_{i\le n} \frac{1}{\sum_i \tau_i^2 I_{i \le N_n}} \mathbb{E}(  \big|\sum_j \tau_j I_{N_{i-1}<j \le N_i}   (I_{Y_i=0}-\frac{1}{2} )\big|^2I_{|\frac{\sum_j \tau_j I_{N_{i-1}<j \le N_i}   (I_{Y_i=0}-\frac{1}{2}) }{\sqrt{\sum_i \tau_i^2 I_{i \le N_n}}}|>\epsilon}|F_{i-1})
\\ & \overset{(b)}{\le} \sum_{i\le N_n} \frac{1}{\sum_i \tau_i^2 I_{i \le N_n}} \tau_{i}^2 P(|\frac{\sum_j \tau_j I_{N_{i-1}<j \le N_i}   (I_{Y_i=0}-\frac{1}{2}) }{\sqrt{\sum_i \tau_i^2 I_{i \le N_n}}}|>\epsilon|F_{i-1})
\\& \le \sum_{i\le N_n} \frac{1}{\sum_i \tau_{i}^2 I_{i \le N_n}} \tau_{i}^2 I_{\frac{\tau_{i}^2}{\sum_i \tau_i^2 I_{i \le N_n}} \ge \epsilon^2} \le \max_i I_{\frac{\tau_{i}^2}{\sum_i \tau_i^2 I_{i \le N_n}} \ge \epsilon^2}
 \le I_{\max_i \frac{\tau_{i}^2}{\sum_i \tau_i^2 I_{i \le N_n}}\ge \epsilon^2}.
\end{split}\end{equation*}
Note that to obtain Equality (a) we used the fact that $\tau_1, \tau_2, \ldots$ are $F_i$ measurable and hence so is $N_n$. To obtain Inequality (b) we used the fact that for a fixed $i$ the difference between $N_{i-1}$ and $N_{i}$ is at most one and also $|I_{Y_i=0}-\frac{1}{2}|<1$. Finally, it is straightforward to see that $P(\max_i \frac{\tau_{i}^2}{\sum_i \tau_i^2 I_{i \le N_n}}\ge \epsilon^2) \rightarrow 0$, which proves \eqref{eq:condthm11}. 

According to Theorem Theorem \ref{Martingale} we have
\[
\frac{1}{\sqrt{\sum_i \tau_i^2 I_{i \le N_n}} } \sum_{i} \tau_i I_{i \le N_n}   (I_{Y_i=0}-\frac{1}{2}) \xrightarrow{d} N(0,\frac{1}{4}).
\]

Hence if 

\[
\sqrt{n} \left( \frac{ \sum_{i \le N_n} \tau_i  I_{Y_i=0} }{n} - \mathbb{E} \Big(\frac{ \sum_{i \le N_n} \tau_i  I_{Y_i=0} }{n}  | \ N_n, \tau_1, \tau_2, \ldots, \tau_{N_n}\Big)\right)
\]converges to a non-degenerate distribution we need: $\sum_i \tau_i^2 I_{i \le N_n}=\Theta(n)$. However, by Cauchy-Swartz we can easily see that: $\sum_i \tau_i^2 I_{i \le N_n} \ge \frac{1}{N_n} ( \sum_{i\le N_n} \tau_i) ^2=\frac{1}{N_n} (n-l_1-left(n) )^2 $. Therefore $\frac{\sum_i \tau_i^2 I_{i \le N_n}}{n}\ge \frac{\frac{1}{N_n} (n-l_1-left(n) )^2}{N_n\times n} \rightarrow \infty$. This contradiction proves that the speed of convergence is slower than $n^{-\frac{1}{2}}$.

\bibliographystyle{unsrt}
\bibliography{../myrefs}

\end{document}